%% file: main.tex
\documentclass[conference]{IEEEtran}
\IEEEoverridecommandlockouts

\input{command}
\graphicspath{{./graphics/}}

\usepackage[firstpage]{draftwatermark}
\SetWatermarkText{\hspace*{6.2in}\raisebox{7.8in}
  {\includegraphics[scale=0.1]{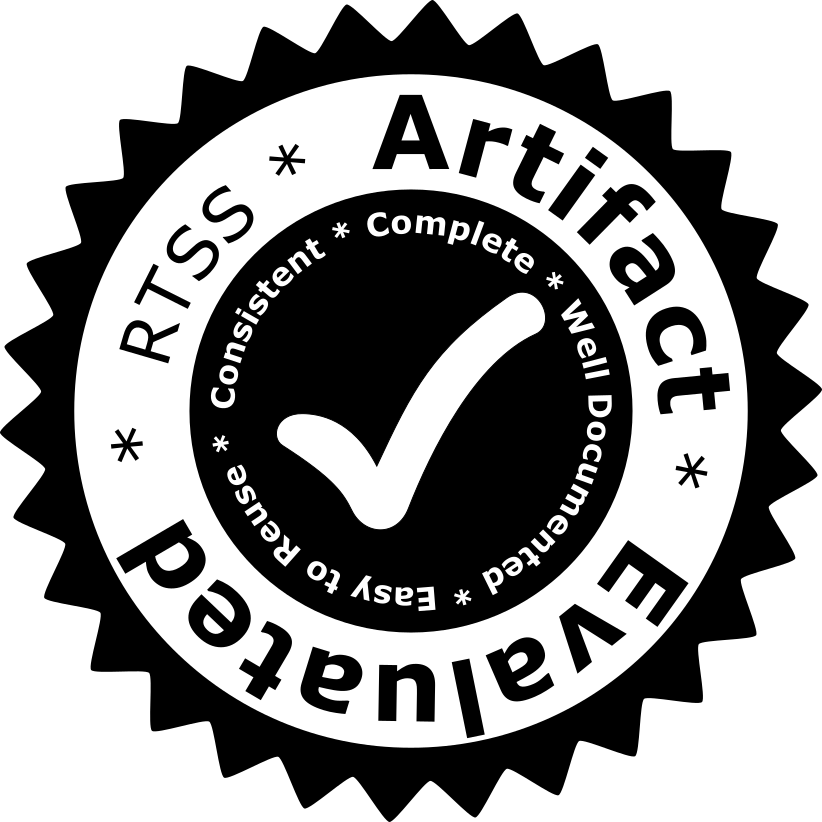}}}
\SetWatermarkAngle{0}
\usepackage{url}

\begin{document}

\title{Bounding the Response Time of DAG Tasks \\ Using Long Paths}

\author{Qingqiang He\textsuperscript{1},
Nan Guan\textsuperscript{2},
Mingsong Lv\textsuperscript{1,3},
Xu Jiang\textsuperscript{3},
Wanli Chang\textsuperscript{4}
\\
\\
\textsuperscript{1}The Hong Kong Polytechnic University, Hong Kong SAR\\
\textsuperscript{2}City University of Hong Kong, Hong Kong SAR\\
\textsuperscript{3}Northeastern University, China \\
\textsuperscript{4}Hunan University / Huawei Technologies, China
\thanks{Corresponding author: Nan Guan, Email: nanguan@cityu.edu.hk.}
\thanks{This paper has passed an Artifact Evaluation process.}
\thanks{
This paper was accepted to 43rd IEEE Real-Time Systems Symposium (RTSS 2022).
}
}

\maketitle

\begin{abstract}
In 1969, Graham developed a well-known response time bound for a DAG task using the total workload and the longest path of the DAG, which has been widely applied to solve many scheduling and analysis problems of DAG-based task systems.
This paper presents a new response time bound for a DAG task using the total workload and the lengths of
\emph{multiple long paths} of the DAG, instead of the \emph{longest} path in Graham's bound. Our new bound theoretically dominates and empirically outperforms Graham's bound.
We further extend the proposed approach to multi-DAG task systems.
Our schedulability test theoretically dominates federated scheduling and outperforms the state-of-the-art by a considerable margin.
\end{abstract}


\section{Introduction}
\label{sec:introduction}
\input{introduction}

\section{Related Work}
\label{sec:related_work}
\input{related_work}

\section{Preliminary}
\label{sec:preliminary}
\input{preliminary}


\section{Response Time Analysis}
\label{sec:analysis}
\input{analysis}

\section{Computing Generalized Path List}
\label{sec:path_sequence}
\input{computation}


\section{Extension to Multi-Task Systems}
\label{sec:extension}
\input{extension}

\section{Evaluation}
\label{sec:evaluation}
\input{evaluation}

\section{Conclusion}
\label{sec:conclusion}
This paper developed a closed-form response time bound using the total workload and the lengths of
multiple relatively long paths of the DAG.
The new bound theoretically dominates and empirically outperforms Graham's bound.
We also extend our result to the scheduling of multi-DAG task systems, which theoretically dominates the original federated scheduling and outperforms the state-of-the-art by a considerable margin.
Currently, the computation of the proposed bound requires the DAG structure of real-time applications.
In the future, we plan to investigate the concept of virtual path further and use virtual paths, instead of paths, to compute the proposed bound.
Since virtual path may be obtained by measuring the execution of real-time applications, it is possible that the proposed bound can be computed without knowing the DAG structure.
Another direction is considering how to integrate the techniques in this paper with other scheduling approaches that do not directly use, but are based on the techniques of Graham's bound.

\section*{Acknowledgment}
This work is supported by the Research Grants Council of Hong Kong (GRF 11208522, 15206221) and the National Natural Science Foundation of China (NSFC 62102072).
The authors also thank the anonymous reviewers for their helpful comments.

\bibliographystyle{IEEEtran}
\bibliography{reference}

\end{document}

%% file: command.tex
\usepackage[caption=false,font=footnotesize]{subfig}
\usepackage{cite}
\usepackage{amsmath,amssymb,amsfonts,amsthm}
\usepackage{bm}
\usepackage{graphicx}
\usepackage{textcomp}
\usepackage{microtype}
\usepackage{csquotes}
\usepackage{xcolor}
\usepackage{comment}
\usepackage{float}
\usepackage{mathtools}
\usepackage[ruled,linesnumbered]{algorithm2e}
\SetKw{Continue}{continue}
\SetKwInOut{Input}{Input}
\SetKwInOut{Output}{Output}

\newtheorem{theorem}{Theorem}
\newtheorem{lemma}{Lemma}
\newtheorem{corollary}{Corollary}
\newtheorem{definition}{Definition}
\newtheorem{example}{Example}
\DeclareMathOperator*{\argmax}{arg\,max}

\newcommand{\gn}[1]{\textcolor[rgb]{1.0,0.00,0.00}{#1}}

\def\BibTeX{{\rm B\kern-.05em{\sc i\kern-.025em b}\kern-.08em
    T\kern-.1667em\lower.7ex\hbox{E}\kern-.125emX}}

\newcommand{\mypath}{\lambda}
\newcommand{\pv}{\pi} 
\newcommand{\tu}{\delta} 

%% file: introduction.tex
This paper studies the response time bounds of DAG (directed acyclic graph) tasks \cite{tang2022real, li2022parallel} under \emph{work-conserving} scheduling (eligible vertices must be executed if there are available cores) on a computing platform with multiple identical cores.
Graham developed a well-known response time bound for a DAG task, using the total workload and the longest path of the DAG \cite{graham1969bounds}.
This is a general result as it applies to any work-conserving scheduling algorithm. Graham's bound serves as the foundation of a large body of work on real-time scheduling and analysis of parallel workload modeled as DAG tasks \cite{melani2015response, li2013outstanding, li2014analysis, jiang2017semi, jiang2020real, baruah2015federatedconditional, ueter2018reservation, wu2021improving, jiang2022unified, wang2022scheduling}.

Intuitively, Graham's bound is derived by constructing an artificial scenario where vertices not in the longest path do \emph{not} execute in parallel with (and thus assumed to interfere with) the execution of the longest path. However, in real execution, many vertices not in the longest path actually can execute in parallel with (and thus do not interfere with) the longest path, so Graham's bound is rather pessimistic in most cases.

In this paper, we develop a more precise response time bound for a DAG task, using the total workload and the lengths of multiple relatively long paths of the DAG, instead of the single longest path in Graham's bound.
The high-level idea is that, using the information of multiple long paths, we can more precisely identify the workload that has to execute in parallel and thus cannot interfere with each other.
It turns out the analysis technique used by Graham's bound based on the abstraction of \emph{critical path} is not enough to
realize the above idea, and we develop new abstractions (e.g., \emph{virtual path} and \emph{restricted critical path})
and new analysis techniques to derive the new response time bound.

Our bound theoretically dominates and empirically outperforms Graham's bound. Evaluation with synthetic workload under various settings shows that our new bound improves Graham's bound largely and reduces the number of cores required by a DAG task significantly.
We also extend our new techniques to the scheduling and analysis of systems consisting of multiple DAG tasks. Our new approach theoretically dominates federated scheduling \cite{li2014analysis} and offers significantly better schedulability than the state-of-the-art as shown in the empirical evaluation.

The rest of this paper is organized as follows.
Section \ref{sec:related_work} reviews related works.
Section \ref{sec:preliminary} defines the DAG task model, the scheduling algorithm and describes Graham's bound that motivates this work.
Section \ref{sec:analysis} presents our new response time bound assuming that a generalized path list is given.
Section \ref{sec:path_sequence} discusses how to compute this generalized path list for a DAG task.
Section \ref{sec:extension} extends our result to the scheduling of multiple DAG tasks.
The evaluation results are reported in Section \ref{sec:evaluation} and Section \ref{sec:conclusion} concludes the paper.

%% file: related_work.tex
Graham developed a well-known response time bound \cite{graham1969bounds}, using the total workload and the length of the longest path for a DAG task.
Graham's bound is based on the work-conserving property: all cores are busy when the critical path (which is the longest path in the worst case) is not executing.

For scheduling one DAG task,
\cite{he2019intra, zhao2020dag, he2021response, zhao2022dag} improved Graham's bound by enforcing certain priority orders among the vertices, so their results are not general to all work-conserving scheduling algorithms.
\cite{voudouris2017timing, chen2019timing} developed scheduling algorithms based on statically assigned vertex execution order, which are no longer work-conserving.
Some work extended Graham's bound to uniform \cite{jiang2017semi}, heterogeneous \cite{han2019response, lin2022type} and unrelated \cite{voudouris2021bounding} multi-core platforms.
Graham's bound has also been extended to other task models, including conditional DAG \cite{melani2015response}, and graph-based model of OpenMP workload \cite{serrano2015timing, wang2017benchmarking, sun2017real, sun2019calculating, sun2021calculating}.
None of these works improves Graham's bound under the same setting (i.e., scheduling a DAG task on a homogeneous multi-core platform by any work-conserving scheduling algorithm) as the original work \cite{graham1969bounds}.
To our best knowledge, this paper is the first work to do so.

For scheduling multiple DAG tasks,
Graham's bound is widely used and many techniques are developed on top of it.
In federated scheduling \cite{li2014analysis}, where each DAG is scheduled independently on a  set of dedicated cores,
Graham's bound is directly applied to the analysis of each individual task.
Later, federated scheduling was generalized to constrained deadline tasks \cite{baruah2015federatedconstrained}, arbitrary deadline tasks \cite{baruah2015federatedarbitrary}, and conditional DAG tasks \cite{baruah2015federatedconditional}.
To address the resource-wasting problem in federated scheduling, a series of federated-based scheduling algorithms \cite{jiang2017semi, ueter2018reservation, jiang2021virtually} were proposed. All these federated scheduling approaches use Graham's bound to compute the number of cores allocated to a DAG task.
In global scheduling,
\cite{melani2015response, fonseca2017improved, fonseca2019schedulability} developed response time analysis techniques for scheduling DAG tasks under Global EDF or Global RM, where Graham's bound is used for the analysis of intra-task interference.
\cite{li2013outstanding, bonifaci2013feasibility, baruah2014improved, jiang2019utilization} proposed schedulability tests for Global EDF or Global RM, which borrow the idea behind the derivation of Graham's bound, although the bound itself is not directly used.
Our work, as a direct improvement of Graham's bound, can potentially be integrated into the above approaches to improve the schedulability of multiple DAG tasks.


%% file: preliminary.tex
\subsection{Task Model}
\label{sec:task}
A parallel real-time task is modeled as a DAG $G = (V, E)$, where $V$ is the set of vertices and $E\subseteq  V \times V$ is the set of edges.
Each vertex $v\in V$ represents a piece of sequentially executed workload with worst-case execution time (WCET) $c(v)$.
An edge $(v_i, v_j)\in E$ represents the precedence relation between $v_i$ and $v_j$, i.e., $v_j$ can  start execution only after vertex $v_i$ finishes its execution.
A vertex with no incoming (outgoing) edges is called a \emph{source vertex} (\emph{sink vertex}). Without loss of generality, we assume that $G$ has exactly one source  (denoted as $v_{src}$), and one sink (denoted as $v_{snk}$).
In case $G$ has multiple source/sink vertices, a dummy source/sink vertex with zero WCET can be added to comply with our assumption.

A \emph{path} is denoted by $\mypath =  (\pv_0,\cdots,\pv_k)$, where $\forall i\in [0,k-1]: (\pv_i, \pv_{i+1})\in E$.
We also use $\mypath$ to denote the set of vertices that are in path $\mypath$.
The length of a path $\mypath$ is defined as $len(\mypath) \coloneqq \sum_{\pv_i\in \mypath}c(\pv_i)$.
A \emph{complete path} is a path $(\pv_0,\cdots,\pv_k)$ such that $\pv_0 = v_{src}$ and $\pv_k = v_{snk}$, i.e., a complete path is a path starting from the source vertex and ending at the sink vertex.
The \emph{longest path} is a complete path with largest $len(\mypath)$ among all paths in $G$, and we use $len(G)$ to denote the length of the longest path.
For any vertex set $V'\subset V$, $vol(V') \coloneqq \sum_{v\in V'}c(v)$.
The volume of $G$ is the total workload in the DAG task, defined as $vol(G) \coloneqq vol(V)=\sum_{v\in V} c(v)$.
If there is an edge $(u,v) \in E$, $u$ is a \emph{predecessor} of $v$. If there is a path in $G$ from $u$ to $v$, $u$ is an \emph{ancestor} of $v$. We use $pre(v)$ and $anc(v)$ to denote the set of predecessors and ancestors of $v$, respectively.

\begin{figure}[t]
\centering
\subfloat[a DAG example]{
    \includegraphics[width=0.37\linewidth]{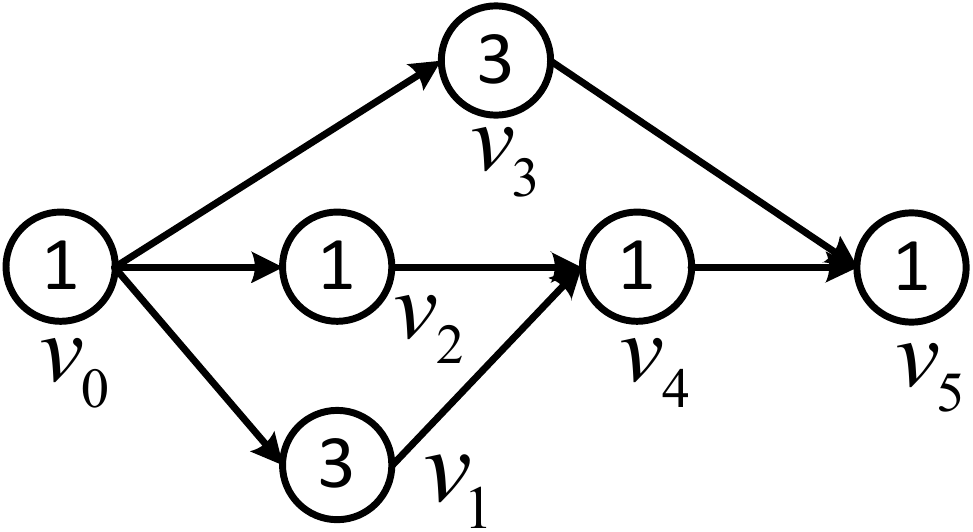}
    \label{fig:dag_example}
}
\hfill
\subfloat[a unit DAG]{
    \includegraphics[width=0.37\linewidth]{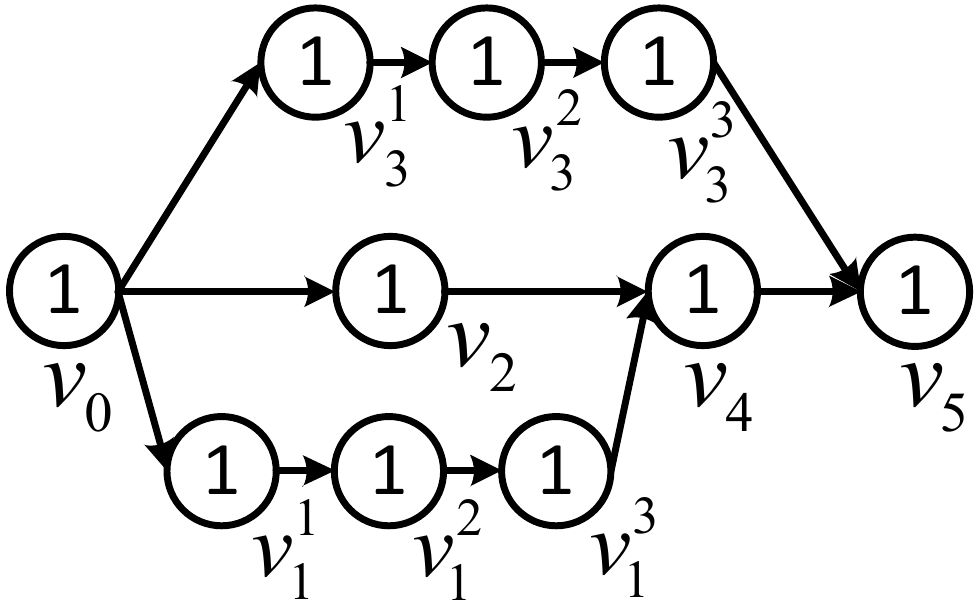}
    \label{fig:unit_dag}
}
\hfill
\subfloat[an execution sequence]{
    \includegraphics[width=0.48\linewidth]{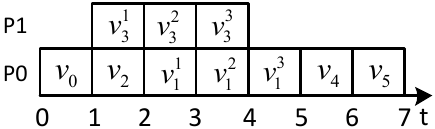}
    \label{fig:execution_sequence}
}
\hfill
\subfloat[another execution sequence]{
    \includegraphics[width=0.42\linewidth]{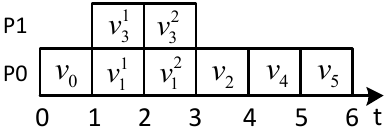}
    \label{fig:workload_unit}
}
\caption{An illustrative example.}
\label{fig:example}
\end{figure}

\begin{example}\label{exp:dag_example}
Fig. \ref{fig:dag_example} shows a DAG task $G$ where the number inside each vertex is its WCET.
$v_0$ and $v_5$ are the source and the sink vertex, respectively.
The longest path is $\mypath=(v_0, v_1, v_4, v_5)$, so $len(G) = len(\mypath)=6$.
For vertex set $V'=\{v_1, v_3\}$, $vol(V')=6$.
The volume of the DAG $vol(G)=10$.
For vertex $v_4$, $pre(v_4)=\{v_1, v_2\}$, $anc(v_4)=\{v_0, v_1, v_2\}$.
\end{example}

\subsection{Runtime Behavior}
\label{sec:runtime}
The vertices of DAG task $G$ are scheduled to execute on a multi-core platform with $m$ identical cores $(P_i)_0^{m-1}$  (which is the compact representation of $(P_0, P_1, \cdots, P_{m-1})$ ).
A vertex $v$ is \emph{eligible} if all its predecessors have finished and thus $v$ can immediately execute if there are available cores.
DAG task $G$ is scheduled by \emph{any} algorithm that satisfies the \emph{work-conserving} property, i.e.,
an eligible vertex must be executed if there are available cores.

At runtime, vertices of $G$ execute at certain time points on certain cores under the decision of the scheduling algorithm.
An \emph{execution sequence} of $G$ describes  which vertex executes on which core at every time point.

Since $c(v)$ is the \emph{worst-case} execution time, some vertices may actually execute for less than their WCETs.
In an execution sequence $\varepsilon$, a vertex $v$ has an execution time $e(v) \in [0, c(v)]$, which is the accumulated executing time of $v$ in $\varepsilon$.
The \emph{start time} $s(v)$ and \emph{finish time} $f(v)$ are the time point when $v$ first starts its execution and completes its execution, respectively.
Note that $e(v)$, $s(v)$ and $f(v)$ are all specific to a certain execution sequence $\varepsilon$, but we do not include $\varepsilon$ in their notations for simplicity.

Without loss of generality, we assume the source vertex of $G$ starts execution at time $0$, so the \emph{response time} of $G$ in an execution sequence equals $f(v_{snk})$.
This paper aims to derive a safe upper bound on the response time $R$ of $G$ in any execution sequence under any work-conserving scheduling.

\begin{example}\label{exp:runtime_example}
For the DAG $G$ in Fig. \ref{fig:dag_example}, suppose $m=2$.
Two possible execution sequences under work-conserving scheduling are shown in Fig. \ref{fig:execution_sequence} and Fig. \ref{fig:workload_unit}
 where $v_1^1$, $v_1^2$ and $v_1^3$ means the execution of first, second and third time unit of $v_1$.
In Fig. \ref{fig:execution_sequence}, every vertex in $G$ executes for its WCET.
In Fig. \ref{fig:workload_unit}, $v_1$ and $v_3$ execute for less than their WCETs and $e(v_1)=e(v_3)=2$.
In Fig. \ref{fig:workload_unit}, the start time and finish time of $v_1$ are $s(v_1)=1$ and $f(v_1)=3$, respectively.
The response times of $G$ for execution sequences in Fig. \ref{fig:execution_sequence} and Fig. \ref{fig:workload_unit} are 7 and 6, respectively.
\end{example}

\subsection{Graham's Bound}
\label{sec:motivation}
For a DAG task under work-conserving scheduling, Graham developed a well-known response time bound \cite{graham1969bounds}.
\begin{theorem}[Graham's Bound \cite{graham1969bounds}]
\label{thm:classic_bound}
The response time $R$ of DAG $G$ scheduled by work-conserving scheduling on $m$ cores is bounded by:
\begin{equation}
\label{equ:classic_bound}
R\le len(G)+\frac{vol(G)-len(G)}{m}
\end{equation}
\end{theorem}

\begin{figure}[t]
  \centering
  \includegraphics[width=0.46\linewidth]{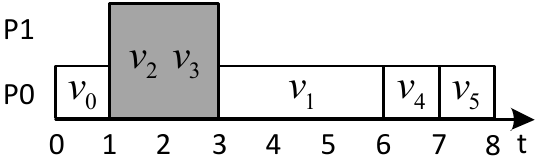}
  \caption{Computation of Graham's bound.}
  \label{fig:classic_bound}
\end{figure}

We use the example in Fig. \ref{fig:example} to illustrate the pessimism in Graham's bound.
By (\ref{equ:classic_bound}), assuming $m=2$, Graham's bound for the DAG task in Fig. \ref{fig:dag_example} is computed as $6+(10-6)/2=8$.
Fig. \ref{fig:classic_bound} illustrates the intuition of its computation.
Workload not in the longest path is considered to be the interference ($v_2$, $v_3$ in this example), which is equally distributed among all cores to calculate the delay to the longest path (the gray area with length $2$ in Fig. \ref{fig:classic_bound}).
However, the workload of $v_3$ must execute sequentially according to the semantics of DAG, which renders this ``equal distribution of interference'' impossible, since $v_3$'s workload of length $3$ cannot fit into an interval of length $2$.
This fact is also illustrated by the execution sequence in Fig. \ref{fig:execution_sequence}, where the real delay to the longest path is $1$, not $2$.
In this paper, we explore insights into the characterization of the execution of a DAG task to address this type of pessimism and derive a tighter response time bound for work-conserving scheduling which analytically dominates Graham's bound.

%% file: analysis.tex
\begin{table}[t]
\centering
\caption{Major Notations Used in the Paper}
\label{tab:notation}
\begin{tabular}{ll}
\hline
\textbf{Notation}   & \textbf{Description} \\
\hline
$\varepsilon$   & an execution sequence \\
$\tu$           & a time unit \\
$c(v)$          & the WCET of vertex $v$ \\
$e(v)$          & the execution time of vertex $v$ \\
$s(v)$          & the start time of vertex $v$ \\
$f(v)$          & the finish time of vertex $v$ \\
$s(\tu)$        & the start time of time unit $\tu$ \\
$f(\tu)$        & the finish time of time unit $\tu$ \\

\hline
$len(\mypath)$  & the length of path $\mypath$ \\
$len(G)$        & the length of the longest path of DAG $G$ \\
$vol(V')$       & the volume (total workload) of vertex set $V'$ \\
$vol(G)$        & the volume of $G$ \\
$pre(v)$        & the set of predecessors of vertex $v$ \\
$anc(v)$        & the set of ancestors of vertex $v$ \\

\hline
$\mypath^*$     & the critical path (Definition \ref{def:critical_path}) \\
$\mypath^+$     & the restricted critical path (Definition \ref{def:restricted_path}) \\
$\omega$        & a virtual path (Definition \ref{def:virtual_path}) \\
$len(\omega)$   & the length of virtual path $\omega$ \\
$(\omega_i)_0^k$        & a virtual path list (Definition \ref{def:virtual_sequence}) \\
$(\mypath_i)_0^k$       & a generalized path list (Definition \ref{def:path_sequence}) \\

\hline
$V'_\varepsilon$        & the projection of $V'$ regarding $\varepsilon$ (Definition \ref{def:projection}) \\
$\mypath_\varepsilon$   & the projection of $\mypath$ regarding $\varepsilon$ (Definition \ref{def:projection}) \\
$\Delta(V')$    & the workload reduction of $V'$ in $\varepsilon$ (Equation \ref{equ:delta}) \\
$\bar{k}$       & a value returned by Algorithm \ref{alg:sequence_computation},  \\
~               & $\bar{k}+1$ is the number of long paths in $G$ \\
\hline
\end{tabular}
\end{table}

This section presents the methodology of deriving a tighter bound for a DAG task.
After introducing Lemma \ref{lem:sequence_bound}, an overview of the analysis method is provided in the end of Section \ref{sec:execution_sequence}.
Major notations used in this paper are summarized in Table \ref{tab:notation}.

\subsection{Analysis on an Execution Sequence}
\label{sec:execution_sequence}
We assume time is discrete and the length of a \emph{time unit} is $1$, which is reasonable, because everything in a digital computer is driven by discrete clocks.
We use $\tu$ to denote a time unit, and $s(\tu)$ and $f(\tu)$ the start time and finish time of $\tu$.

\textbf{Unit DAG.}
We transform each vertex of $G$ into a series of \emph{unit vertices}.
The WCET $c(v)$ of each unit vertex $v$ is $1$ (i.e., a time unit). The resulting DAG is a \emph{unit DAG}.
For example, for the DAG in Fig. \ref{fig:dag_example}, the unit DAG is shown in Fig. \ref{fig:unit_dag}.
Since a time unit cannot be further divided, in an execution sequence, the execution time $e(v)$ of a unit vertex $v$ is either $1$ or $0$.
A unit DAG is also a DAG and notations introduced for DAGs are also applicable to unit DAGs.
Unless explicitly specified, $G = (V, E)$ is a unit DAG in Section \ref{sec:analysis}.
Note that the unit DAG is merely used as an auxiliary concept for the proofs, and our results (i.e., Theorem \ref{thm:our_bound} and Algorithm \ref{alg:sequence_computation}) do \emph{not} need to really transfer the original DAG
into a unit DAG.

Our analysis focuses on an \emph{arbitrary} execution sequence $\varepsilon$ of unit DAG $G$.
Unless explicitly specified, the following definitions and discussions are all for execution sequence $\varepsilon$.

\begin{definition}[Critical Predecessor \cite{sun2020computing}]\label{def:critical_predecessor}
In an execution sequence, vertex $u$ is a \emph{critical predecessor} of vertex $v$, if
\begin{equation}\label{equ:critical_predecessor}
u=\argmax_{u_i \in pre(v)} \{f(u_i)\}
\end{equation}
\end{definition}

\begin{definition}[Critical Path \cite{he2019intra}]\label{def:critical_path}
In an execution sequence, a \emph{critical path} $\mypath^* = (\pv_0, \cdots, \pv_k)$ ending at vertex $v$ is a path satisfying the following two conditions.
\begin{itemize}
  \item $\pv_0=v_{src} \land \pv_k=v$;
  \item $\forall \pv_i \in \mypath^* \setminus \{\pv_0\}$, $\pv_{i-1}$ is a critical predecessor of $\pv_i$.
\end{itemize}
\end{definition}
As a special case of the above definition, a critical path of $G$ in an execution sequence is a critical path ending at $v_{snk}$.
For any vertex $v \ne v_{src}$, since a critical predecessor of $v$ must exist, we can always find the critical path ending at $v$.
The critical path is specific to an execution sequence of $G$. A critical path of $G$ in an execution sequence
is \emph{not} necessarily the longest path of $G$.

\begin{example}
For the execution sequence in Fig. \ref{fig:execution_sequence}, a critical path of $G$ is $(v_0, v_1^1, v_1^2, v_1^3, v_4, v_5)$.
In Fig. \ref{fig:workload_unit}, a critical path of $G$ is ($v_0, v_2, v_4, v_5$), which is not the longest path of $G$.
\end{example}

\begin{lemma}\label{lem:conserving_property}
In an execution sequence under work-conserving scheduling on $m$ cores, for any vertex $v$ and its critical predecessor $u$, all $m$ cores are busy in time interval $[f(u),s(v)]$.
\end{lemma}
\begin{proof}
Since $u$ is a critical predecessor of $v$, by Definition \ref{def:critical_predecessor}, $v$ is eligible at $f(u)$.
If some core is idle in $[f(u),s(v)]$, it contradicts the fact that the scheduling is work-conserving.
\end{proof}

In an execution sequence, the execution time of some vertices in a vertex set $V'$ may be less than their WCETs.
In the following, we introduce notations to describe workloads of vertices in an execution sequence.

\begin{definition}[Projection]\label{def:projection}
In a execution sequence $\varepsilon$, the \emph{projection} of a vertex set $V'$ is defined as
\begin{equation}\label{equ:worklaod_V}
V'_\varepsilon \coloneqq \{v \in V'| e(v)= 1 \textrm{ in } \varepsilon\}
\end{equation}
and the projection of a path $\mypath$ is defined as
\begin{equation}\label{equ:worklaod_path}
\mypath_\varepsilon \coloneqq \{\pv_i \in \mypath | e(\pv_i)= 1 \textrm{ in } \varepsilon\}
\end{equation}
\end{definition}

Intuitively, a projection $V'_\varepsilon$ is a vertex set including vertices from $V'$ whose execution time is not 0 in $\varepsilon$.
As a special case, $V_\varepsilon$ is the projection of the vertex set $V$ of the DAG $G$ in $\varepsilon$.
By definition,
\begin{equation}\label{equ:volume_V}
vol(V'_\varepsilon)=\sum_{v\in V'_\varepsilon} c(v) =\sum_{v\in V'} e(v)
\end{equation}
\begin{equation}\label{equ:lenght_path}
len(\mypath_\varepsilon)= \sum_{\pv_i \in \mypath_\varepsilon} c(\pv_i) = \sum_{\pv_i\in \mypath} e(\pv_i)
\end{equation}

\begin{example}\label{exp:projection}
Consider the execution sequence $\varepsilon$ in Fig. \ref{fig:workload_unit}.
For vertex set $V'=\{v_1^1, v_1^2, v_1^3, v_3^1, v_3^2, v_3^3 \}$, in $\varepsilon$, the execution times of some vertices from $V'$ is 0. $V'_\varepsilon=\{v_1^1, v_1^2, v_3^1, v_3^2 \}$.
The volume of $V'_\varepsilon$ is $vol(V'_\varepsilon)=4$, while $vol(V')=6$.
The total workload of $G$ in $\varepsilon$ is $vol(V_\varepsilon)=8$.
For path $\mypath=(v_0, v_1^1, v_1^2, v_1^3, v_4, v_5)$, $\mypath_\varepsilon=(v_0, v^1_1, v^2_1, v_4, v_5)$,
$len(\mypath_\varepsilon)=5$, while $len(\mypath)=6$.
\end{example}

Now we introduce a key concept \emph{virtual path} to describe the \emph{sequentially} executed workload in an execution sequence.

\begin{definition}[Virtual Path]\label{def:virtual_path}
In an execution sequence, a \emph{virtual path} $\omega$ is a set of vertices executing in different time units.
\end{definition}
Same as path, the length of a virtual path $\omega$ is defined as $len(\omega) \coloneqq \sum_{v \in \omega} c(v)$.
Since virtual path $\omega$ is defined regarding an execution sequence $\varepsilon$, a virtual path does not include vertices whose execution time in $\varepsilon$ is 0. In $\varepsilon$, $\forall v \in \omega$, $e(v)=c(v)=1$.

All the vertices in a virtual path of $\varepsilon$ do not execute in parallel in $\varepsilon$. In other words, a virtual path is a sequentially executed workload in $\varepsilon$.
Note that vertices in a virtual path of $\varepsilon$ may execute in parallel in another execution sequence $\varepsilon'$ of $G$.
A path is always a virtual path in any execution sequence. However, a virtual path is not necessarily a path.

\begin{example}
In Fig. \ref{fig:execution_sequence}, $\omega_0=(v_3^1, v_3^2, v_3^3, v_4)$ is a virtual path and $len(\omega_0)=4$. $\omega_0$ is also a virtual path in the execution sequence in Fig. \ref{fig:workload_unit}.
In Fig. \ref{fig:workload_unit}, $\omega_1=(v_1^1, v_3^2, v_2)$ is a virtual path and $len(\omega_1)=3$. $\omega_1$ is not a virtual path in the execution sequence in Fig. \ref{fig:execution_sequence}, where $v_1^1$ and $v_3^2$ execute in parallel.
\end{example}

\begin{definition}[Virtual Path List]\label{def:virtual_sequence}
A virtual path list is a set of disjoint virtual paths $(\omega_i)_0^k$ ($k \ge 0$), i.e.,
$$\forall i, j \in [0, k],\ \omega_i \cap \omega_j = \varnothing$$
\end{definition}
Here $(\omega_i)_0^k$ is the compact representation of $(\omega_0, \cdots, \omega_k)$.
Slightly abusing the notation, we also use $(\omega_i)_0^k$ to denote the set of vertices that are in some $\omega_i$ ($i \in [0, k]$).

\begin{figure}[t]
\centering
\begin{minipage}{.17\textwidth}
  \centering
  \includegraphics[width=1\linewidth]{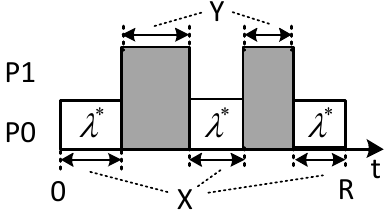}\\
  \caption{$X$ and $Y$.}
  \label{fig:interference}
\end{minipage}
\hfill
\begin{minipage}{.27\textwidth}
  \centering
  \includegraphics[width=0.85\linewidth]{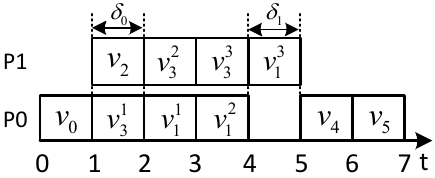}\\
  \caption{Workload swapping.}
  \label{fig:renaming}
\end{minipage}
\end{figure}

For critical path $\mypath^*$ of execution sequence $\varepsilon$, we define
\begin{itemize}
  \item $X$: time interval during which $\exists \pv_i \in \mypath^*$, $\pv_i$ is executing;
  \item $Y$: time interval before $f(v_{snk})$ during which $\forall \pv_i \in \mypath^*$, $\pv_i$ is \emph{not} executing.
\end{itemize}
In this paper, a time interval is not necessarily continuous.
Fig. \ref{fig:interference} illustrates the definitions of $X$ and $Y$.
As an example, in Fig. \ref{fig:workload_unit}, $X=[0, 1] \cup [3, 6]$, $Y=[1, 3]$, and $|X|=4$, $|Y|=2$.

\textbf{Workload Swapping.} Next, we introduce a procedure called \emph{workload swapping} which transforms an execution sequence into another one.
The purpose of workload swapping is to put the workload into cores in a way that is more convenient to present our analysis.
For an execution sequence $\varepsilon$, workload swapping includes the following two operations:
\begin{itemize}
  \item swap two vertices executing on two cores in the same time unit;
  \item move a vertex to the same time unit of another idle core (i.e., ``swap'' a vertex with an ``idle block'' on another core in the same time unit).
\end{itemize}
Applying the above procedure to $\varepsilon$ generates a new execution sequence $\varepsilon'$.
For a vertex $v$, the core on which $v$ is executing in $\varepsilon$ may be different from that of $\varepsilon'$.
Since the start time and finish time of each vertex in $\varepsilon'$ remain the same as $\varepsilon$,
the timing behaviors of $\varepsilon$ and $\varepsilon'$ are actually the same.
Therefore, the response time $R$ does not change.
Time intervals, such as $X$, $Y$, and the critical path do not change either. A virtual path in $\varepsilon$ is still a virtual path in $\varepsilon'$ and virtual path lists do not change. For a path $\mypath$, $\mypath_\varepsilon= \mypath_{\varepsilon'}$.

\begin{example}
For the execution sequence in Fig. \ref{fig:execution_sequence}, a new execution sequence generated by workload swapping (swapping $v_2$ and $v^1_3$ in time unit $\tu_0$, and swapping $v^3_1$ to core $P_1$ in $\tu_1$) is shown in Fig. \ref{fig:renaming}.
\end{example}

\begin{lemma}\label{lem:sequence_bound}
$\varepsilon$ is an execution sequence  of DAG $G$ under work-conserving scheduling on $m$ cores,
$\mypath^*$ is the critical path of $\varepsilon$.
$\mypath^*_\varepsilon$ is the projection of $\mypath^*$ in $\varepsilon$.
Given a virtual path list $(\omega_i)_0^k$ ($k \in [0, m-1]$) in $\varepsilon$ where $\omega_0 = \mypath^*_\varepsilon$, the response time $R$ of $\varepsilon$ is bounded by:
\begin{equation}\label{equ:sequence_bound}
R\le len(\mypath^*_\varepsilon)+\frac{vol(V_\varepsilon)-\sum_{i=0}^{k} len(\omega_i)}{m-k}
\end{equation}
\end{lemma}
\begin{proof}
First, we claim the following properties.
\begin{itemize}
  \item [~~~~A.] $R=f(v_{snk})=|X|+|Y|$ (by definitions of $X$ and $Y$);
  \item [~~~~B.] $|X|=len(\mypath^*_\varepsilon)$ (by the definition of $X$);
  \item [~~~~C.] during time interval $Y$, all cores are busy with vertices from $V_\varepsilon \setminus \mypath^*_\varepsilon$ (by Lemma \ref{lem:conserving_property});
  \item [~~~~D.] $(\omega_i)_1^k \subseteq V_\varepsilon \setminus \mypath^*_\varepsilon$ (by Definition \ref{def:virtual_sequence}, virtual paths in a virtual path list are disjoint).
\end{itemize}

Recall that $vol(V_\varepsilon)$ is the total workload of $G$ in $\varepsilon$.
By workload swapping, we swap each virtual path $\omega_i$ to core $P_i$ ($i = 1, ..., k$), which generates a new execution sequence $\varepsilon'$ with the same $X$, $Y$, and $R$ as $\varepsilon$.
Let $W$ denote the set of vertices executing in other $m-k$ cores (i.e., $P_0$, $(P_i)_{k+1}^{m-1}$) during $Y$ in $\varepsilon'$, then by Property C, we have
\begin{equation}\label{equ:bound_Y}
|Y|=\frac{vol(W)}{m-k}
\end{equation}
Since we swap all $(\omega_i)_1^k$ into cores $(P_i)_1^{k}$, we have
$$W \subseteq I, \textrm{~where~} I \coloneqq V_\varepsilon \setminus \mypath^*_\varepsilon \setminus (\omega_i)_1^k$$
Therefore, by Property A, B and (\ref{equ:bound_Y}), we have
\begin{equation}\label{equ:workload_bound}
R=|X|+|Y|=len(\mypath^*_\varepsilon)+\frac{vol(W)}{m-k} \le len(\mypath^*_\varepsilon)\!+\!\frac{vol(I)}{m-k}
\end{equation}
By the definition of $I$ and Property D, we have
\begin{align*}
vol(I)  &= vol(V_\varepsilon)-len(\mypath^*_\varepsilon)-\sum_{i=1}^{k} len(\omega_i) \\
        &= vol(V_\varepsilon)-\sum_{i=0}^{k} len(\omega_i)
\end{align*}
which, together with (\ref{equ:workload_bound}), completes the proof.
\end{proof}

\textbf{Method Overview.}
Lemma \ref{lem:sequence_bound} gives a response time bound for a particular execution sequence.
However, Lemma \ref{lem:sequence_bound} cannot be directly used to upper-bound the response time of the DAG as it requires values available only when the execution sequence is given, which is unknown in offline analysis.
Therefore, in the following, we will bound these execution-sequence-specific values using static information of the DAG.
We rewrite (\ref{equ:sequence_bound}) as
\begin{equation}\label{equ:sequence_bound_rewrite}
R\le len(\mypath^*_\varepsilon)+\frac{vol(V_\varepsilon)- len(\omega_0)}{m-k} -\frac{\sum_{i=1}^{k} len(\omega_i)}{m-k}
\end{equation}
In Section \ref{sec:restricted}, we introduce a new abstraction called restricted critical path, and investigate its properties.
In Section \ref{sec:sequential_workload}, using the results of Section \ref{sec:restricted}, we lower-bound $\sum_{i=1}^{k} len(\omega_i)$ (i.e., lower-bound $\frac{\sum_{i=1}^{k} len(\omega_i)}{m-k}$), and then
in Section \ref{sec:bound_dag}, we upper-bound $len(\mypath^*_\varepsilon)+\frac{vol(V_\varepsilon)- len(\omega_0)}{m-k}$. Combining them yields an upper bound of the RHS (right-hand side) of (\ref{equ:sequence_bound_rewrite}).

\subsection{Restricted Critical Path}
\label{sec:restricted}
This subsection introduces a key concept \emph{restricted critical path}, which  is essentially a critical path identified within a subset of vertices in $G$.
In line with the critical path, the restricted critical path is to further characterize the execution behavior of a DAG task with the awareness of multiple long paths in the execution.
We first generalize the concept of path.
\begin{definition}[Generalized Path]\label{def:generalized_path}
A generalized path $\mypath=(\pv_0,\cdots, \pv_k)$ is a set of vertices such that $\forall i\in [0,k-1]$, there is a path $\mypath_i$ starting at $\pv_i$ and ending at $\pv_{i+1}$. In particular, a vertex set containing only one vertex is a generalized path.
\end{definition}

Intuitively, a generalized path \enquote{skips} some vertices in a path so the vertices in a generalized path may not directly connect to each other.
For example, in Fig. \ref{fig:dag_example}, $(v_0, v_2, v_4, v_5)$ is a path, while $(v_0, v_2, v_5)$ is a generalized path.
Same as path, the length of a generalized path $\mypath$ is defined as $len(\mypath) \coloneqq \sum_{\pv_i\in \mypath} c(\pv_i)$.

The relationship among \emph{path}, \emph{generalized path} and \emph{virtual path} can be summarized as
\[\textrm{path} \subseteq \textrm{generalized path} \subseteq \textrm{virtual path}\]
A path must be a generalized path; a generalized path is not necessarily a path. A generalized path must be a virtual path in any execution sequence of $G$; a virtual path in an execution sequence is not necessarily a generalized path.
Path and generalized path share a property: vertices in a path or a generalized path always execute sequentially in any execution sequence of $G$. However, this is not true for virtual path: vertices in a virtual path in one execution sequence may not  execute sequentially in another execution sequence.

\begin{definition}[Generalized Path List]\label{def:path_sequence}
A \emph{generalized path list} is a virtual path list, in which each element is a generalized path. A generalized path list is denoted as $(\mypath_i)_0^k$, where each $\mypath_i$ is a generalized path.
\end{definition}

\begin{definition}[Restricted Critical Path]\label{def:restricted_path}
For an execution sequence $\varepsilon$ and a generalized path list $(\mypath_i)_0^k$, $k \in [0, m-1]$, the \emph{restricted critical path}  $\mypath^+ = (\pv_0, \cdots, \pv_j)$ is a generalized path satisfying (\ref{equ:last_vertex}), (\ref{equ:restricted_path}) and (\ref{equ:first_vertex}):
\begin{equation}\label{equ:last_vertex}
\pv_j = \argmax_{u\in (\mypath_i)_0^k} \{f(u)\}
\end{equation}
\begin{equation}\label{equ:restricted_path}
\forall \pv_i \in \mypath \setminus \{\pv_0\}: \pv_{i-1} = \argmax_{u \in anc(\pv_i) \cap (\mypath_i)_0^k} \{f(u)\}
\end{equation}
\begin{equation}\label{equ:first_vertex}
anc(\pv_0) \cap (\mypath_i)_0^k = \varnothing
\end{equation}
\end{definition}
Intuitively, the concept of restricted critical path is obtained by applying the concept of critical path to vertices in $(\mypath_i)_0^k$.
The restricted critical path ends at the last finishing vertex in $(\mypath_i)_0^k$ (Equation \ref{equ:last_vertex}).
After $\pv_i$ is identified, we identify $\pv_{i-1}$ as the last finishing vertex of ancestors of $\pv_i$ in $(\mypath_i)_0^k$ (Equation \ref{equ:restricted_path}).
This recursive procedure stops until the last identified vertex's ancestor is not in $(\mypath_i)_0^k$ (Equation \ref{equ:first_vertex}).
If $(\mypath_i)_0^k$ includes all vertices of the graph, the restricted critical path of an execution sequence degrades to the critical path of that execution sequence.

\begin{example}\label{exp:restricted_path}
In Fig. \ref{fig:dag_workload}, the number inside vertices is to identify the vertex, not the WCET. The WCET of each vertex is 1.
Fig. \ref{fig:sequence_workload} shows an execution sequence $\varepsilon$, and the execution time of each vertex in $\varepsilon$ is 1. The number of cores is 3.
Let  $\mypath_0=(v_0, v_1, v_2, v_3, v_{12})$, which is the longest path. Let $\mypath_1=(v_4, v_5, v_6)$, and $\mypath_2=(v_7, v_8, v_9)$.
$(\mypath)_0^2$ is a generalized path list.
In $\varepsilon$, with respect to $(\mypath)_0^2$, we can identify a restricted critical path $\mypath^+=(v_0, v_4, v_5, v_6, v_{12})$ (the green vertices).
\end{example}

Now we transform an execution sequence $\varepsilon$ into a \enquote{regular} form by workload swapping.

\begin{definition}[Regular Execution Sequence]\label{def:regular_sequence}
Given an execution sequence $\varepsilon$ and a generalized path list $(\mypath_i)_0^k$ ($k \in [0, m-1]$), we transform $\varepsilon$ into a \emph{regular execution sequence} $\varepsilon'$ regarding $(\mypath_i)_0^k$ via workload swapping by the following two rules:
\begin{enumerate}
  \item  swap  $\mypath_i$ to core $P_i$ for each $i \in [0, k]$;
  \item swap other vertices into cores with a smaller index as much as possible.
\end{enumerate}
\end{definition}

\begin{lemma}\label{lem:idle_property}
Let $v$ be an arbitrary vertex of $G$, $\tu$ be a time unit before $v$ starts (i.e., $f(\tu) \le s(v)$) during which some core is idle, then there exists an ancestor of $v$ executing in $\tu$.
\end{lemma}
\begin{proof}
We prove by contradiction. Assume all ancestors of $v$ do not execute in $\tu$.
Let $\mypath^* = (\pv_0, \cdots, \pv_{k})$ be the critical path ending at $v$ (i.e., $v = \pv_{k}$).
By our assumption, $\tu$ must be in some time interval when $\mypath^*$ is not executing.
By Lemma \ref{lem:conserving_property}, we know all cores are busy in $\tu$, which contradicts that some cores are idle in $\tu$. The lemma is proved.
\end{proof}

\begin{figure}[t]
\centering
\subfloat[a DAG task]{
    \includegraphics[width=0.32\linewidth]{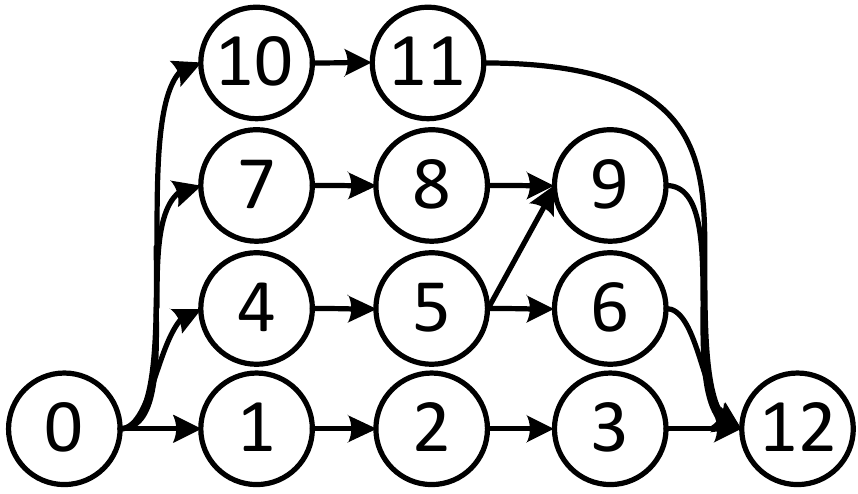}
    \label{fig:dag_workload}
}
\hfil
\subfloat[an execution sequence]{
    \includegraphics[width=0.41\linewidth]{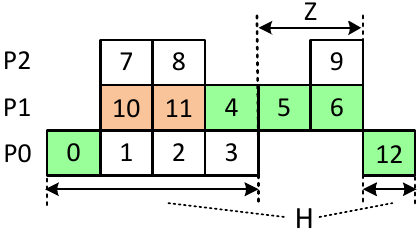}
    \label{fig:sequence_workload}
}
\caption{An example of restricted critical path.}
\label{fig:workload}
\end{figure}

\begin{lemma}\label{lem:restricted}
$(\mypath_i)_0^k$, $k \in [0, m-1]$, is a generalized path list where $\mypath_0$ is the longest path of $G$.
$\varepsilon$ is a regular execution sequence regarding $(\mypath_i)_0^k$.
$\mypath^+$ is the restricted critical path of $(\mypath_i)_0^k$ in $\varepsilon$.
There exists a virtual path $\eta$ in $\varepsilon$ satisfying all the following three conditions:
\begin{enumerate}
  \item [(i)] $\forall v \in \eta$, $v \notin (\mypath_i)_0^k$;
  \item [(ii)] $\forall v \in \eta$, $v$ executes on $(P_i)_0^{k}$;
  \item [(iii)] $len(\mypath^+_\varepsilon)+ len(\eta)=f(v_{snk})$.
\end{enumerate}
\end{lemma}
\begin{proof}
Let $\mypath^+ = (\pv_0, \cdots, \pv_j)$. Since $\mypath_0$ is the longest path, $\mypath_0$ includes $v_{src}$ and $v_{snk}$, so $(\mypath_i)_0^k$ includes $v_{src}$ and $v_{snk}$.
By Definition \ref{def:restricted_path}, we can identify a restricted critical path $\mypath^+$ of $(\mypath_i)_0^k$ in $\varepsilon$ satisfying $\pv_0=v_{src}$ and $\pv_j=v_{snk}$.
For each $h \in [1, j]$, we examine the vertices in each time unit in $[f(\pv_{h-1}), s(\pv_h)]$. For each $\tu \in [f(\pv_{h-1}), s(\pv_{h})]$, we will prove that there is a vertex $v \notin (\mypath_i)_0^k$ executing on $(P_i)_0^{k}$ in $\tu$. We prove this by contradiction, assuming that such $v$ does not exist. There are two cases.
\begin{enumerate}
  \item In $\tu$, on $(P_i)_0^{k}$, all $k+1$ vertices are from $(\mypath_i)_0^k$.\\
    Suppose that $\pv_{h}$ executes on core $P_r$ ($r \in [0, k]$) and let $u$ denote the vertex executing on $P_r$ in $\tu$, so $u \in (\mypath_i)_0^k$.
    Since $u$ executes in $\tu$, we have $f(u) \ge f(\tu)$.
    Since $\varepsilon$ is a regular execution sequence, if a vertex in $(\mypath_i)_0^k$ executes on $P_r$, this vertex must be in $\mypath_r$, so both $\pv_{h}$ and $u$ are in $\mypath_r$,
    which implies that either $\pv_{h}$ is an ancestor of $u$ or the other way around.
    Moreover, since $u$ executes in $\tu \in [f(\pv_{h-1}), s(\pv_{h})]$, we know $u$ is an ancestor of $\pv_h$.
    Therefore $\exists u \in anc(\pv_{h}) \cap (\mypath_i)_0^k$ and $f(u) \ge f(\tu) >f(\pv_{h-1})$.

  \item In $\tu$, on $(P_i)_0^{k}$, less than $k+1$ vertices are from $(\mypath_i)_0^k$.\\
    By assumption, all vertices executing in $\tu$ are in $(\mypath_i)_0^k$, thus at least one core is idle in $\tu$.
    Therefore, by Lemma \ref{lem:idle_property}, we know there exists an ancestor of $\pv_h$, denoted by $u$, executing in $\tu$, which implies $f(u) \ge f(\tu)$.
    By assumption, all vertices executing in $\tu$ are in $(\mypath_i)_0^k$, so $u \in (\mypath_i)_0^k$.
    Therefore, $\exists u \in anc(\pv_h) \cap (\mypath_i)_0^k$ and $f(u) \ge f(\tu) >f(\pv_{h-1})$.
\end{enumerate}

In summary, for both cases we have proved that $\exists u \in anc(\pv_{h}) \cap (\mypath_i)_0^k$ and $f(u) >f(\pv_{h-1})$.
On the other hand, by Definition \ref{def:restricted_path}, in particular (\ref{equ:restricted_path}),
$\pv_{h-1}$ has the maximum finish time among vertices in $anc(\pv_{h}) \cap (\mypath_i)_0^k$,
which contradicts the existence of $u$.
Therefore, our assumption must be false, i.e., $\forall h \in [1, j]$, $\forall \tu \in [f(\pv_{h-1}), s(\pv_{h})]$, we can find such $v$ satisfying $v \notin (\mypath_i)_0^k$ and $v$ executes on $(P_i)_0^{k}$.

We collect such vertices in each time unit $\tu \in [f(\pv_{h-1}), s(\pv_{h})]$, for each $h \in [1, j]$.
These vertices form a virtual path $\eta$.
$len(\eta)$ equals the total length of time intervals $[f(\pv_{0}), s(\pv_{1})], [f(\pv_{1}), s(\pv_{2})], \cdots, [f(\pv_{j-1}), s(\pv_{j})]$, i.e., all
time intervals before $f(v_{snk})$ during which $\mypath^+_\varepsilon$ is not executing.
Therefore, $len(\mypath^+_\varepsilon)+ len(\eta)=f(v_{snk})$.
The lemma is proved.
\end{proof}

\begin{example}\label{exp:restricted}
$(\mypath_i)_0^2$ is the generalized path list in Example \ref{exp:restricted_path}.
Execution sequence $\varepsilon$ in Fig. \ref{fig:sequence_workload} is a regular execution sequence regarding $(\mypath)_0^2$.
In $\varepsilon$, the virtual path $\eta$ identified by Lemma \ref{lem:restricted} is $\eta=(v_{10}, v_{11} )$ (the brown vertices in Fig. \ref{fig:sequence_workload}).
\end{example}

By the properties of restricted critical path in Lemma \ref{lem:restricted}, the sequentially executed workload (i.e., the volume of virtual paths) for different complete paths can be analyzed quantitatively in Section \ref{sec:sequential_workload}.

\subsection{Lower-bounding $\sum_{i=1}^{k} len(\omega_i)$}
\label{sec:sequential_workload}
This subsection is the most technically challenging part of this work. We develop constructive proofs to derive the desired lower bound.
The bound in (\ref{equ:sequence_bound}) holds for an arbitrary virtual path list.
Therefore, we only need to construct a particular virtual path list for which $\sum_{i=1}^{k} len(\omega_i)$ can be lower-bounded.

For the longest path $\mypath$ and execution sequence $\varepsilon$, we define
\begin{itemize}
  \item $H$: time interval during which $\exists \pv_i \in \mypath$, $\pv_i$ is executing;
  \item $Z$: time interval before $f(v_{snk})$ during which $\forall \pv_i \in \mypath$, $\pv_i$ is \emph{not} executing.
\end{itemize}
By definition, $|H|=len(\mypath_\varepsilon)$, $|H|+|Z|=f(v_{snk})$.
Recall that $\mypath_\varepsilon$ is the projection of $\mypath$ regarding $\varepsilon$ (i.e., eliminating the vertices in $\mypath$ with zero execution time in $\varepsilon$).
Note that the definitions of $H$ and $Z$ are different from the definitions of $X$ and $Y$. $X$ is the time interval in which the \emph{critical} path is executing, whereas $H$ is the time interval in which the \emph{longest} path is executing.
As an example, in Fig. \ref{fig:workload_unit}, $H=[0, 3] \cup [4, 6]$, $Z=[3, 4]$, and $|H|=5$, $|Z|=1$.

$H$ and $Z$ are introduced to construct a new virtual path list based on generalized path list $(\mypath_i)_0^k$ where $\mypath_0$ is the longest path of $G$.
Recall the requirements for a virtual path list:
1) vertices in a virtual path must execute sequentially (i.e., execute in different time units);
2) the virtual paths in a virtual path list must be disjoint.
These two requirements should be kept in mind when constructing the desired virtual path list.

For a vertex set $V'$ and an execution sequence $\varepsilon$, we define
\begin{equation}\label{equ:delta}
\Delta(V') \coloneqq vol(V')-vol(V'_\varepsilon)
\end{equation}
$\Delta(V')$ represents the amount of workload reduction of $V'$ in $\varepsilon$. Obviously, if $V'_0 \subseteq V'_1$, then $\Delta(V'_0) \le \Delta(V'_1)$.
For example, in Fig. \ref{fig:workload_unit}, for $V'=\{v_1^1, v_1^2, v_1^3, v_3^1, v_3^2, v_3^3\}$, $\Delta(V')=6-4=2$ (see Example \ref{exp:projection}).

\begin{lemma}\label{lem:degree}
$W_0$ is a set of vertices. In execution sequence $\varepsilon$, $(\omega_i)_0^k$ is virtual path list satisfying $\bigcup_{i \in [0, k]} \omega_i = W_0$. $\forall v \in W_0$, $\forall u \notin W_0$ satisfying $u$ is in the same time unit as $v$. $W_1 \coloneqq (W_0 \setminus v) \cup u$. There exists a virtual path list $(\omega'_i)_0^k$ such that $\bigcup_{i \in [0, k]} \omega'_i = W_1$.
\end{lemma}
\begin{proof}
Suppose $v$ is from $\omega_j$. $\omega'_j \coloneqq (\omega_j \setminus v) \cup u$. Since $v$ and $u$ are in the same time unit, $\omega'_j$ is a virtual path.
$\forall i \in [0, k]$, $i \ne j$, $\omega'_i \coloneqq \omega_i$. Since $u \notin W_0$, obviously, $(\omega'_i)_0^k$ is virtual path list and $\bigcup_{i \in [0, k]} \omega'_i = W_1$.
\end{proof}

\begin{figure}[t]
  \centering
  \includegraphics[width=0.9\linewidth]{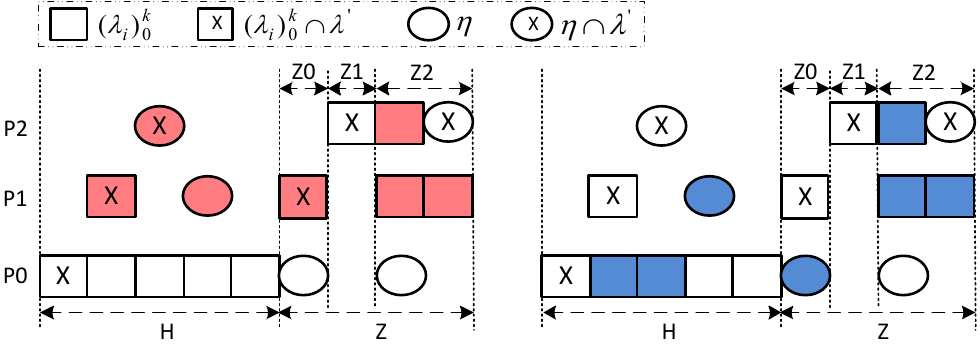}\\
  \caption{The construction in Lemma \ref{lem:workload}. Regular execution sequence $\varepsilon'$ is depicted in the figure. The red vertices represent $W$; the blue vertices represent $W'$. Vertices in $(\mypath_i)_0^k$ are depicted as squares; vertices in $\eta$ as circles. Vertices in $\mypath'$ are labeled with ``X''. Other vertices are not depicted.
  }
  \label{fig:lem_workload}
\end{figure}

Lemma \ref{lem:degree} will be used in Lemma \ref{lem:workload_exist}.
Now, we start to construct the desired virtual path list.
Again, $\varepsilon$ is an arbitrary execution sequence under analysis. $\mypath$ is the longest path of $G$.
$(\mypath_i)_0^k$ ($k \in [0, m-1]$) is a generalized path list where $\mypath_0=\mypath$.
For any complete path $\mypath'$ of $G$, we construct a virtual path list $(\omega_i)_0^k$ where $\omega_0=\mypath'_{\varepsilon}$.
$\mypath'_{\varepsilon}$ is the projection of $\mypath'$ regarding $\varepsilon$.

After constructing $\omega_0=\mypath'_{\varepsilon}$, in the following, we will construct $(\omega_i)_1^k$.
The construction of $(\omega_i)_1^k$ will be conducted on two levels. First, we construct a vertex set $W'$ which
satisfies $W'=\bigcup_{i \in [1, k]}\omega_i$, i.e., includes all vertices in $(\omega_i)_1^k$.
This will be done in Algorithm \ref{alg:construction}. Second, in Lemma \ref{lem:workload_exist}, we will prove that we can construct $(\omega_i)_1^k$ using the vertices in $W'$.

We construct $(\omega_i)_1^k$ based on $(\mypath_i)_1^k$.
Let $\theta \coloneqq (\mypath_{i\varepsilon})_1^k \cap \mypath'_\varepsilon$,
where
 $\mypath_{i\varepsilon}$ denotes the projection of $\mypath_{i}$ regarding $\varepsilon$, i.e., the set of vertices from $\mypath_{i}$ whose execution time is not 0 in $\varepsilon$ (Definition \ref{def:projection}).
Since $\mypath'$ is a path, $\theta$ is a generalized path.
We transform $\varepsilon$ into a regular execution sequence $\varepsilon'$ regarding $(\mypath_i)_0^k$.
Let $\mypath^+$ be the restricted critical path of $(\mypath_i)_0^k$.
By Lemma \ref{lem:restricted}, there is a virtual path $\eta$ in $\varepsilon'$ satisfying:
(i) $\forall v \in \eta$, $v \notin (\mypath_i)_0^k$;
(ii) $\forall v \in \eta$, $v$ executes on $(P_i)_0^{k}$;
(iii) $len(\mypath^+_\varepsilon)+ len(\eta)=f(v_{snk})$.
By (i), since $\eta \cap (\mypath_i)_0^k= \varnothing$, $\theta \subseteq (\mypath_i)_0^k$, we have $\eta \cap \theta=\varnothing$.
We divide $Z$ into the following time intervals:
\begin{itemize}
  \item $Z_0$: $\theta$ is executing and $\eta$ is executing;
  \item $Z_1$: $\theta$ is executing and $\eta$ is not executing;
  \item $Z_2$: $\theta$ is not executing.
\end{itemize}
Recall that a time interval is in general not continuous.
Obviously, $|Z|=|Z_0|+|Z_1|+|Z_2|$.
Let $\eta_H$ denote the set of vertices which is from $\eta$ and is in $H$.
Let $\theta_{Z_1}$ denote the set of vertices which is from $\theta$ and is in $Z_1$. We have $len(\theta_{Z_1})=|Z_1|$.
$W \coloneqq (\mypath_{i\varepsilon})_1^k \cup \eta_H \setminus \theta_{Z_1}$.

\begin{algorithm}[h]
    \caption{Constructing $W'$}\label{alg:construction}
    \DontPrintSemicolon
    \Input{$W$, $\mypath'_{\varepsilon}$}
    \Output{$W'$}
    $W' \leftarrow W$ \\
    \ForEach{$v \in W \cap \mypath'_{\varepsilon}$}{
        \uIf{$v$ \textup{is in} $H$}{
            $u \leftarrow$ the vertex which is from $\mypath$ and is in the same time unit as $v$ \\
            $W' \leftarrow (W' \setminus v) \cup u$
        }
        \ElseIf{$v$ \textup{is in} $Z_0$}{
            $u \leftarrow$ the vertex which is from $\eta$ and is in the same time unit as $v$ \\
            $W' \leftarrow (W' \setminus v) \cup u$
        }
    }
\end{algorithm}

Next, based on $W$, we construct a vertex set $W'$ using Algorithm \ref{alg:construction}.
See Fig. \ref{fig:lem_workload} for illustration.
For time interval $H$, in Line 3-5, we replace all vertices of $\mypath'_{\varepsilon}$ (squares or circles labeled with \enquote{X} in Fig. \ref{fig:lem_workload}) with vertices of the longest path $\mypath$ (squares in $P_0$).
For time interval $Z_0$, in Line 6-8, we replace all vertices of $\mypath'_{\varepsilon}$ with vertices of $\eta$ (circles).
Note that Algorithm \ref{alg:construction} is only for the constructive proof, not really needed
for computing our new response time bound.
We can prove the following properties of $W'$.

\begin{lemma}\label{lem:workload_exist}
There exists a virtual path list $(\omega_i)_1^k$ satisfying $\bigcup_{i \in [1, k]} \omega_i = W'$.
\end{lemma}
\begin{proof}
By (ii) of Lemma \ref{lem:restricted}, in $\varepsilon'$, $\eta$ executes on $(P_i)_0^{k}$ and $\eta \cap (\mypath_i)_0^k= \varnothing$.
Also in $\varepsilon'$ during $H$, $\mypath_0$ is executing, which means that $\eta_H$ executes on $(P_i)_1^{k}$.
Therefore, In $\varepsilon'$, all vertices of $W$ execute in $(P_i)^k_1$, which means that there exists a virtual path list $(\omega'_i)_1^k$ satisfying $\bigcup_{i \in [1, k]} \omega'_i = W$.
By Lemma \ref{lem:degree}, after each iteration of the loop in Line 2-9 of Algorithm \ref{alg:construction},
there exists a virtual path list $(\omega_i)_1^k$ satisfying $\bigcup_{i \in [1, k]} \omega_i = W'$.
\end{proof}

\begin{lemma}\label{lem:workload_disjoint}
$W' \cap \mypath'_{\varepsilon} = \varnothing$.
\end{lemma}
\begin{proof}
In $H$, all vertices from $\mypath'_{\varepsilon}$ are excluded from $W'$ in Line 3-5 of Algorithm \ref{alg:construction}.
In $Z_0$, all vertices from $\mypath'_{\varepsilon}$ are excluded from $W'$ in Line 6-8 of Algorithm \ref{alg:construction}.
In $Z_1$, $W \cap \theta_{Z_1} =\varnothing$, which means $W' \cap \theta_{Z_1} =\varnothing$.
In $Z_2$, by definition of $Z_2$, $\theta$ is not executing.
\end{proof}

\begin{lemma}\label{lem:workload_volume}
$vol(W') \ge \sum_{i=1}^{k} len(\mypath_i) -\Delta(V)$.
\end{lemma}
\begin{proof}
By (iii) of Lemma \ref{lem:restricted}, $len(\mypath^+_{\varepsilon})+len(\eta)=f(v_{snk})$.
By definitions of $H$, $Z$, $len(\mypath_{\varepsilon})+|Z|=f(v_{snk}) $.
Since $\mypath$ is the longest path of $G$, $len(\mypath) \ge len(\mypath^+) \ge len(\mypath^+_{\varepsilon})$.
Therefore,
\begin{align*}
len(\eta)-|Z| &=len(\mypath_{\varepsilon})-len(\mypath^+_{\varepsilon}) \\
&=len(\mypath)-\Delta(\mypath)-len(\mypath^+_{\varepsilon}) \ge -\Delta(\mypath)
\end{align*}

Let $\eta_{Z_0}$, $\eta_{Z_1}$, $\eta_{Z_2}$ denote the set of vertices which are from $\eta$ and are in $Z_0$, $Z_1$, $Z_2$, respectively. We have $\eta=\eta_H \cup \eta_{Z_0}\cup \eta_{Z_1}\cup \eta_{Z_2}$. By definitions of $Z_0$, $Z_1$, $Z_2$, $len(\eta_{Z_0})=|Z_0|$, $len(\eta_{Z_1})=0$, $len(\eta_{Z_2}) \le |Z_2|$. We have $len(\eta_H) \ge len(\eta)-|Z_0|-|Z_2|$.
Obviously, $len(\theta_{Z_1}) = |Z|-|Z_0|-|Z_2|$.
Therefore,
$$len(\eta_H)-len(\theta_{Z_1}) \ge len(\eta)-|Z| \ge -\Delta(\mypath)$$
By (\ref{equ:delta}), $\sum_{i=1}^{k} len(\mypath'_{i\varepsilon})= \sum_{i=1}^{k} len(\mypath_i)-\Delta((\mypath_i)_1^k) $.
\begin{align*}
vol(W') &=vol(W)=\sum_{i=1}^{k} len(\mypath_{i\varepsilon})+len(\eta_H)-len(\theta_{Z_1}) \\
& \ge \sum_{i=1}^{k} len(\mypath_i)-\Delta((\mypath_i)_1^k)-\Delta(\mypath) \\
& = \sum_{i=1}^{k} len(\mypath_i)-\Delta((\mypath_i)_0^k) \ge \sum_{i=1}^{k} len(\mypath_i)-\Delta(V)
\end{align*}
\end{proof}

\begin{lemma}\label{lem:workload}
$\varepsilon$ is an execution sequence. $\mypath$ is the longest path of $G$.
$(\mypath_i)_0^k$, $k \in [0, m-1]$, is a generalized path list where $\mypath_0=\mypath$.
For any complete path $\mypath'$ of $G$, there is a virtual path list $(\omega_i)_0^k$ where $\omega_0=\mypath'_{\varepsilon}$, satisfying the following condition.
\begin{equation}\label{equ:workload}
\sum_{i=1}^{k} len(\omega_i) \ge \sum_{i=1}^{k} len(\mypath_i)-\Delta(V)
\end{equation}
\end{lemma}
\begin{proof}
In the above construction, we have $W'$.
By Lemma \ref{lem:workload_exist}, there exists a virtual path list $(\omega_i)_1^k$ satisfying $\bigcup_{i \in [1, k]} \omega_i = W'$.
By Lemma \ref{lem:workload_disjoint}, $W' \cap \mypath'_{\varepsilon} = \varnothing$ which means
$(\omega_i)_1^k$ and $\mypath'_{\varepsilon}$ are disjoint.
Therefore, $(\omega_i)_0^k$ where $\omega_0=\mypath'_{\varepsilon}$ is a virtual path list.
By Lemma \ref{lem:workload_volume},
we have $\sum_{i=1}^{k} len(\omega_i) =vol(W') \ge \sum_{i=1}^{k} len(\mypath_i)-\Delta(V)$.
The lemma is proved.
\end{proof}

The main idea in the proof of Lemma \ref{lem:workload} is to construct $(\omega_i)_1^k$ using $(\mypath_i)_1^k$. $\mypath'_\varepsilon$ is $\omega_0$ and vertices in $(\mypath_i)_1^k \setminus \theta$ are used for $(\omega_i)_1^k$. However, since $\theta$ is in $\mypath'$ and $\mypath'_\varepsilon= \omega_0$, vertices in $\theta$ cannot be used for $(\omega_i)_1^k$ (recall that virtual paths in a virtual path list are disjoint).
Therefore, to construct $(\omega_i)_1^k$ using $(\mypath_i)_1^k$ is to replace vertices in $\theta$ using vertices that are not in $\mypath'$.
These vertices used to replace $\theta$ are from $\mypath_0$ and $\eta$.
We use the following example to explain this.

\begin{example}
$(\mypath_i)_0^2$ is the generalized path list in Example \ref{exp:restricted_path}.
Let $\mypath'=(v_0, v_4, v_5, v_9, v_{12})$.
Since in Fig. \ref{fig:sequence_workload} each vertex executes for its WCET, $\Delta(V)=0$.
The virtual path $\eta$ identified by Lemma \ref{lem:restricted} is $\eta=(v_{10}, v_{11} )$ (the brown vertices in Fig. \ref{fig:sequence_workload}). $\theta=(v_4, v_5, v_9)$. $|Z|=|Z_1|=2$.
$v_3$ is to replace $v_4$, and $v_{10}$, $v_{11}$ are to replace $v_5$, $v_9$.
Then a new virtual path list $(\omega)_0^2$ is constructed, where $\omega_0=\mypath'$, $\omega_1=(v_{10}, v_{11}, v_3, v_6)$, and $\omega_2= (v_7, v_8)$.
We have $vol(\omega_1 \cup \omega_2)=6 \ge vol(\mypath_1 \cup \mypath_2)=6$.
\end{example}

%

\subsection{Bound for the DAG Task}
\label{sec:bound_dag}
For concise presentation, we define a function
\begin{equation}\label{equ:B_function}
B(x, y, z) \coloneqq x(1-\frac{1}{m-k})+\frac{y-z}{m-k}
\end{equation}
$B(x, y, z)$ is monotonically increasing with respect to $x$ and $y$, and decreasing with respect to $z$.

Using $B$ function, Lemma \ref{lem:sequence_bound} can be rewritten as
\begin{align*}
R &\le len(\mypath^*_\varepsilon)+\frac{vol(V_\varepsilon)-\sum_{i=0}^{k} len(\omega_i)}{m-k} \nonumber \\
  &=   len(\mypath^*_\varepsilon)(1-\frac{1}{m-k})+\frac{vol(V_\varepsilon)-\sum_{i=1}^{k} len(\omega_i)}{m-k} \\
  &=   B(len(\mypath^*_\varepsilon), vol(V_\varepsilon), \sum_{i=1}^{k} len(\omega_i))
\end{align*}

\begin{lemma}\label{lem:dag_bound}
Given a generalized path list $(\mypath_i)_0^k$ ($k \in [0, m-1]$) where $\mypath_0$ is the longest path of $G$,
the response time $R$ of DAG $G$ scheduled by work-conserving scheduling on $m$ cores is bounded by:
\begin{equation}\label{equ:dag_bound}
R\le len(G)+\frac{vol(G)-\sum_{i=0}^{k} len(\mypath_i)}{m-k}
\end{equation}
\end{lemma}
\begin{proof}
Let $\varepsilon$ be an arbitrary execution sequence of $G$ under work-conserving scheduling. Let $\mypath^*$ denote the critical path of $\varepsilon$.
By Lemma \ref{lem:workload}, for $\mypath^*$, there is a virtual path list $(\omega_i)_0^k$ where $\omega_0=\mypath^*_\varepsilon$ satisfying
$\sum_{i=1}^{k} len(\omega_i) \ge \sum_{i=1}^{k} len(\mypath_i)-\Delta(V)$.
By Lemma \ref{lem:sequence_bound}, the response time $R$ of $\varepsilon$ satisfies
\begin{align*}
R &\le B(len(\mypath^*_\varepsilon), vol(V_\varepsilon), \sum_{i=1}^{k} len(\omega_i)) \\
  &\le B(len(\mypath^*_\varepsilon), vol(G)-\Delta(V), \sum_{i=1}^{k} len(\mypath_i)-\Delta(V)) \\
  & =  B(len(\mypath^*_\varepsilon), vol(G), \sum_{i=1}^{k} len(\mypath_i)) \\
  &\le B(len(\mypath_0), vol(G), \sum_{i=1}^{k} len(\mypath_i))
\end{align*}

Recall that $\mypath_0$ is the longest path and $len(\mypath_0)=len(G)$. The lemma is proved.
\end{proof}

\begin{theorem}\label{thm:our_bound}
Given a generalized path list $(\mypath_i)_0^k$ ($k \in [0, m-1]$) where $\mypath_0$ is the longest path of $G$, the response time $R$ of DAG $G$ scheduled by work-conserving scheduling on $m$ cores is bounded by:
\begin{equation}\label{equ:our_bound}
R\le \min \limits_{j \in [0, k]} \left\{ len(G)+\frac{vol(G)-\sum_{i=0}^{j} len(\mypath_i)}{m-j} \right\}
\end{equation}
\end{theorem}
\begin{proof}
By Lemma \ref{lem:dag_bound}, we know that for each $j \in [0, k]$, $len(G)+\frac{vol(G)-\sum_{i=0}^{j} len(\mypath_i)}{m-j}$ is an upper bound of $R$. Therefore, the minimum of these bounds also upper-bounds $R$.
\end{proof}

The derivation procedure of the bound in Theorem \ref{thm:our_bound}  is based on unit DAGs.
However, the result of Theorem \ref{thm:our_bound} can be directly applied to the original DAG.
For a DAG $G$ and its corresponding unit DAG $G^u$, $vol(G)=vol(G^u)$; for a path $\mypath$ of $G$, and its corresponding path $\mypath^u$ of $G^u$, $len(\mypath)=len(\mypath^u)$.
For example, in Fig. \ref{fig:dag_example}, let $\mypath=(v_0, v_1)$; in Fig. \ref{fig:unit_dag}, the corresponding path is $\mypath^u=(v_0, v^1_1, v^2_1, v^3_1)$. Obviously, $len(\mypath)=len(\mypath^u)=4$.
Therefore, the result of Theorem \ref{thm:our_bound} directly applies to the original DAG $G$.

Note that by our analysis, the bound in (\ref{equ:our_bound}) is still safe when some vertices execute for less than their WCETs.

\begin{corollary}\label{cor:bound_domination}
The bound in (\ref{equ:our_bound}) dominates Graham's bound.
\end{corollary}
\begin{proof}
$len(\mypath_0)=len(G)$.
Let $j=0$ in (\ref{equ:our_bound}), we have $len(G)+\frac{vol(G)-len(G)}{m}$, which is Graham's bound. Therefore,
the bound in (\ref{equ:our_bound}) is less than or equal to Graham's bound.
\end{proof}

The derivation of our bound only depends on the work-conserving property. Same as Graham's bound, our bound is valid for any work-conserving scheduling algorithm, regardless of whether it is preemptive or non-preemptive, priority-based or other rule-based.

%% file: computation.tex
To compute the bound  (\ref{equ:our_bound}) in Theorem \ref{thm:our_bound}, a generalized path list $(\mypath_i)_0^k$ should be given in advance. This section presents how to compute the generalized path list.
Note that any generalized path list is qualified to compute the bound in (\ref{equ:our_bound}). The target here is to find a generalized path list to make the bound as small as possible.
As shown in (\ref{equ:our_bound}), $(\mypath_i)_0^k$ with larger volume and smaller $k$ leads to a smaller bound, which will guide our algorithm for computing the generalized path list.
Note that we work on the original DAG, instead of the unit DAG, to compute the generalized path list.
Again, unit DAG is only an auxiliary concept used in the proofs, and is not needed when computing the proposed response time bound.

\begin{definition}[Residue Graph]\label{def:residue_graph}
Given a generalized path $\mypath$ of graph $G = (V, E)$, the residue graph $res(G, \mypath)= (V, E)$ is defined as:
\begin{itemize}
  \item if $v \in \mypath$, the WCET of $v$ in $res(G, \mypath)$ is $0$;
  \item if $v \in V\setminus \mypath$, the WCET of $v$ in $res(G, \mypath)$ is $c(v)$.
\end{itemize}
\end{definition}

The residue graph $res(G, \mypath)$ has the same vertex set and edge set as $G$. However, the WCETs of vertices in $res(G, \mypath)$ and $G$ are different: the WCETs of vertices in path $\mypath$ are set to 0 in $res(G, \mypath)$.
With the concept of residue graph, the generalized path list is computed by Algorithm \ref{alg:sequence_computation}.
The input of Algorithm \ref{alg:sequence_computation} $G=(V, E)$ is a DAG (\emph{not} a unit DAG).
In Line 4, vertices with zero WCET are removed from  $\mypath_i$, which ensures that there are no common vertices among different $\mypath_i$.
The $(\mypath_i)_0^{\bar{k}}$ computed by Algorithm \ref{alg:sequence_computation} is a generalized path list with $\mypath_0$ being the longest path of $G$.

\begin{algorithm}[t]
    \caption{Computing Generalized Path List}\label{alg:sequence_computation}
    \DontPrintSemicolon
    \Input{$G = (V, E)$}
    \Output{$(\mypath_i)_0^{\bar{k}}$}
     $G' \leftarrow G$; $i \leftarrow 0$\\
     \While{$vol(G') \neq 0$}{
        $\mypath_i \leftarrow$ the longest path of $G'$\\
        $\mypath_i \leftarrow \mypath_i \setminus \{v \in \mypath_i| c(v)\ \mathrm{of}\ G'\ \mathrm{is}\ 0 \}$\\
        $G' \leftarrow res(G', \mypath_i)$; $i \leftarrow i+1$\\
     }
\end{algorithm}

\textbf{Complexity.}
In Algorithm \ref{alg:sequence_computation}, the while-loop can execute no more than $|V|$ times. For pseudo-codes in Line 2, 4, and 5, the time complexity is $O(V)$; In Line 3, the time complexity of computing the longest path in a DAG is $O(V+E)$. In summary, the time complexity of Algorithm \ref{alg:sequence_computation} is $O(V(V+E))$.

The $\bar{k}$ of $(\mypath_i)_0^{\bar{k}}$ returned by Algorithm \ref{alg:sequence_computation} only depends on the parameters of $G$. The $k$ of $(\mypath_i)_0^k$ in Theorem \ref{thm:our_bound} relates to both the DAG and the core number $m$.
Therefore, to compute the bound in (\ref{equ:our_bound}), we let $k \in [0, \min(\bar{k}, m-1)]$.
In theory, any $k \in [0, \min(\bar{k}, m-1)]$ is valid to compute (\ref{equ:our_bound}).
In reality, it can be easily seen that the bound in (\ref{equ:our_bound}) with a larger $k$ is no larger than that of a smaller $k$ with respect to the same generalized path list.
Therefore, we simply use $k=\min(\bar{k}, m-1)$ to compute (\ref{equ:our_bound}).

Note that in Algorithm \ref{alg:sequence_computation}, the loop will continue until $vol(G')=0$, which means that for the computed generalized path list, the following is true: $vol(G)=\sum_{i=0}^{\bar{k}} len(\mypath_i)$. This fact will be used in the proof of Theorem \ref{thm:new_core}.

\begin{example}\label{exp:path_sequence}
For the DAG in Fig. \ref{fig:dag_example}, the generalized path list computed by Algorithm \ref{alg:sequence_computation} is $\mypath_0=(v_0, v_1, v_4, v_5)$, $\mypath_1=(v_3)$, and $\mypath_2=(v_2)$,
whose lengths are $6$, $3$ and $1$, respectively.
If $m=2$, the computed bound is $6+(10-6-3)/(2-1)=7$, smaller than Graham's bound, which is $8$.
\end{example}

%% file: extension.tex
In this section, we extend our result to the scheduling of multi-DAG task systems.
We model a DAG $G$ as a tuple $\mathcal{P}$:
\begin{equation}\label{equ:new_model}
\mathcal{P} \coloneqq \langle C, (L_i)_0^{\bar{k}} \rangle
\end{equation}
where $C \coloneqq vol(G)$ is the volume of $G$, and $(L_i)_0^{\bar{k}}$
is a list of numbers, where each $L_i$ equals the length $len(\mypath_i)$
of a generalized path computed  by Algorithm \ref{alg:sequence_computation}.
We define $L \coloneqq L_0 =len(G)$ is the length of the longest path of $G$.
Using these new notations, (\ref{equ:our_bound}) is simplified to be
\begin{equation}\label{equ:new_bound}
R\le \min \limits_{j \in [0, k]} \left\{ L+\frac{C-\sum_{i=0}^{j} L_i}{m-j} \right\}
\end{equation}
where $k = \min(\bar{k}, m-1)$ as discussed in Section \ref{sec:path_sequence}.
Each sporadic parallel task is represented as $(\mathcal{P}, D, T)$ where $\mathcal{P}=\langle C, (L_i)_0^{\bar{k}} \rangle$; $D$ is the relative deadline and $L \le D$; $T$ is the period. We consider constrained deadline, i.e., $D \le T$.

We schedule the multi-DAG system by
the widely-used federated scheduling
approach \cite{li2014analysis}, which is simple to implement and has good guaranteed real-time performance.
In federated scheduling, each heavy task (tasks with $C \ge D$) is assigned and executes exclusively on $m$ cores under a work-conserving scheduler, where $m$ is computed by (\ref{equ:original_core}).
\begin{equation}\label{equ:original_core}
m = \left \lceil \frac{C-L}{D-L} \right \rceil
\end{equation}
The light tasks (tasks with $C < D$) are treated as sequential sporadic tasks and are scheduled on the remaining cores
by sequential multiprocessor scheduling algorithms such as global EDF \cite{baruah2007techniques} or partitioned EDF \cite{baruah2005partitioned}.

To apply our response time bound to federated scheduling, the only extra effort is to decide the number of cores to be allocated to each heavy task, i.e., the minimum $m$ so that the response time bound in (\ref{equ:new_bound}) is no larger than the deadline. This is essentially the same as applying Graham's bound in (\ref{equ:classic_bound}) to federated scheduling to support multi-DAG systems.

\begin{theorem}\label{thm:new_core}
A parallel real-time task
$(\mathcal{P}, D, T)$ where $\mathcal{P}=\langle C, (L_i)_0^{\bar{k}} \rangle$ and $C \ge D \ge L$ is schedulable on $m$ cores where
\begin{equation}\label{equ:new_core}
m = \min_{j \in [0, \bar{k}]} \{ m(j) \}
\end{equation}
\begin{equation*}
m(j) \coloneqq
\begin{cases}
\left\lceil \frac{ C-\sum_{i=0}^{j} L_i }{D-L} \right\rceil +j  & j < \bar{k} \land D>L \\
\bar{k}+1 & j = \bar{k}
\end{cases}
\end{equation*}
\end{theorem}
\begin{proof}
In (\ref{equ:new_bound}), we know $\forall j \in [0, \bar{k}]$ and $j \le m-1$, $L+\frac{C-\sum_{i=0}^{j} L_i}{m-j}$ is a bound on the response time of the task.
The computed core number $m$ should guarantee that the response time bound is no larger than the deadline $D$.
Therefore, $m$ should satisfy both (\ref{equ:new_core_mj}) and (\ref{equ:new_core_less}).
\begin{equation}\label{equ:new_core_mj}
m \ge j+1
\end{equation}
\begin{equation}\label{equ:new_core_less}
L+\frac{C-\sum_{i=0}^{j} L_i}{m-j} \le D
\end{equation}

Also note that by Algorithm \ref{alg:sequence_computation}, if $j < \bar{k}$, then $C > \sum_{i=0}^{j} L_i$; if $j = \bar{k}$, then $C = \sum_{i=0}^{j} L_i$.
There are two cases.

(1) $j < \bar{k} \land D > L$.
By (\ref{equ:new_core_less}), we have
$m \ge \frac{ C-\sum_{i=0}^{j} L_i }{D-L} +j$.
Since the core number is an integer and we want $m$ to be small as much as possible, we have
\begin{equation}\label{equ:new_core_j}
m =\left\lceil \frac{ C-\sum_{i=0}^{j} L_i }{D-L} \right\rceil +j
\end{equation}
As $j < \bar{k}$, $C > \sum_{i=0}^{j} L_i$, $m$ computed by (\ref{equ:new_core_j}) satisfies (\ref{equ:new_core_mj}).

(2) $j=\bar{k}$.
Since $j = \bar{k}$, by Algorithm \ref{alg:sequence_computation}, $C = \sum_{i=0}^{j} L_i$. Note that $D \ge L$. Therefore, (\ref{equ:new_core_less}) satisfies trivially.
By (\ref{equ:new_core_mj}), $m \ge j+1$. Also, we want $m$ to be small as much as possible, so in this case $m = j+1= \bar{k} +1$.

In summary, $\forall j \in [0, \bar{k}]$, $m(j)$ is a valid $m$ in the sense that (\ref{equ:new_core_mj}) and (\ref{equ:new_core_less}) are satisfied.
Therefore, if $m$ is the minimum of all $m(j)$, the response time bound is no larger than the deadline.
\end{proof}

Note that if $D=L$, by Case (2) of the above proof, a valid $m=\bar{k} +1$ can also be computed.
This result is deduced by our theory and is intuitive at the same time. Since the DAG task has $\bar{k} +1$ generalized paths and generalized paths are sequentially executed workloads in any execution sequences, if the allocated number of cores is $\bar{k} +1$,
these $\bar{k} +1$ generalized paths cannot interfere with each other at all.
Therefore, the response time of the DAG task will be no larger than the length of the longest path, so the deadline $D=L$ will not be missed.

\begin{corollary}\label{cor:test_domination}
The schedulability test of federated scheduling with our response time bound
in (\ref{equ:our_bound}) dominates the original schedulability test in \cite{li2014analysis} using Graham's bound.
\end{corollary}
\begin{proof}
It suffices to show that the core number computed by (\ref{equ:new_core}) is no larger than that of (\ref{equ:original_core}), i.e.,
\begin{equation}\label{equ:core_domination}
\min_{j \in [0, \bar{k}]} \{ m(j) \} \le \left \lceil \frac{C-L}{D-L} \right \rceil
\end{equation}
Let $j=0$ in the LHS (left-hand side) of (\ref{equ:core_domination}). We have
$\left \lceil \frac{C-L}{D-L} \right \rceil$, which is the RHS of (\ref{equ:core_domination}).
The corollary holds.
\end{proof}

Besides being directly used in federated scheduling, Graham's bound also
contributes the idea behind its analysis techniques to the analysis
of other scheduling approaches, such as global scheduling \cite{melani2016schedulability}. Similarly, the analysis technique of our
new response time bound also has the potential to be applied to improve the
scheduling and analysis of other scheduling approaches such as global scheduling, which will be studied in our future work.

%% file: evaluation.tex
This section evaluates the performance of our proposed methods.
We conduct experiments of scheduling both single-DAG systems and multi-DAG systems using randomly generated task graphs.

\subsection{Evaluation of Single-DAG Systems}
\label{sec:evaluation_one}
\begin{figure}[t]
\centering
\subfloat[core number]{
    \includegraphics[width=0.49\linewidth]{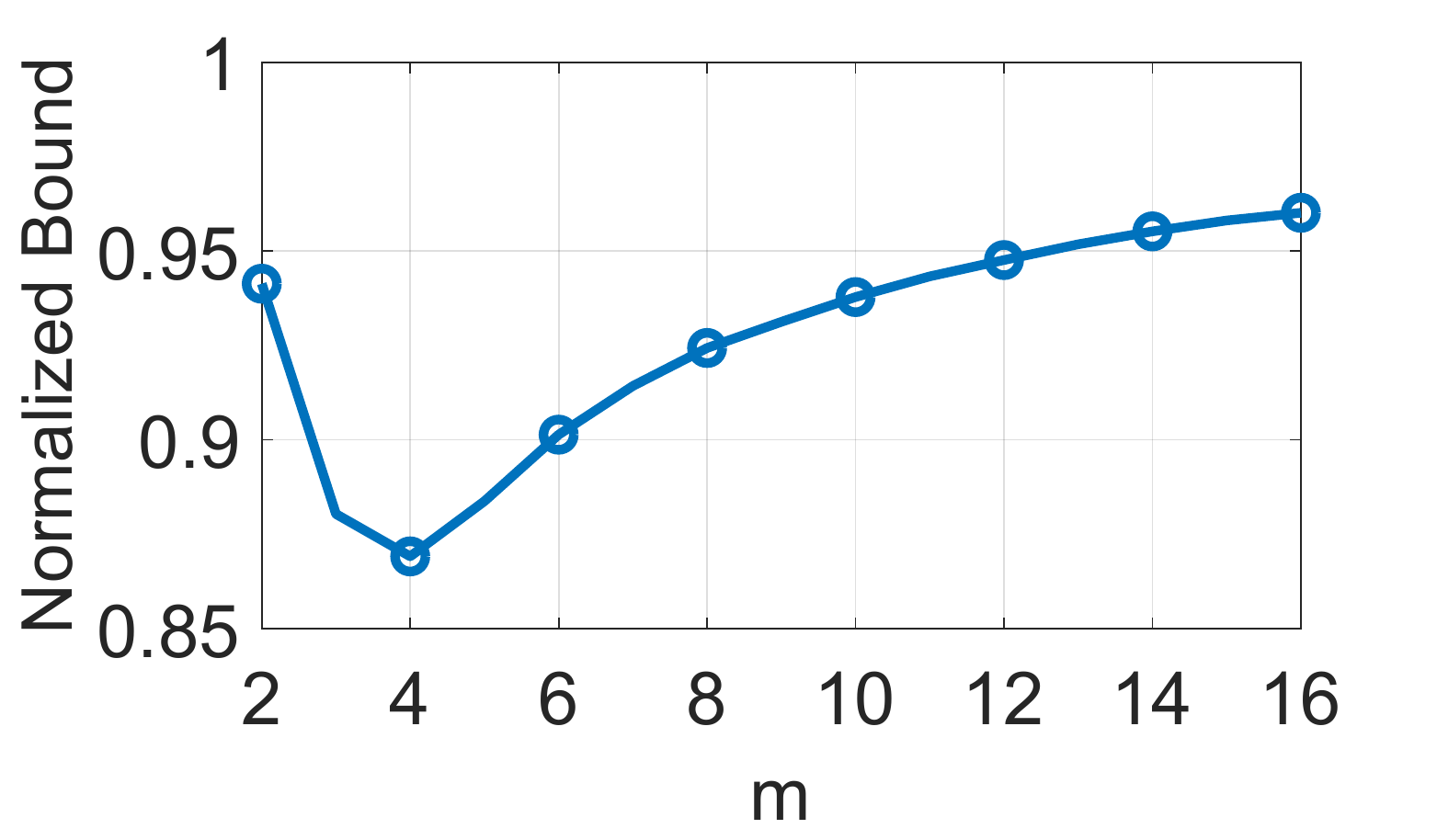}
    \label{fig:m}
}
\subfloat[parallelism factor]{
    \includegraphics[width=0.49\linewidth]{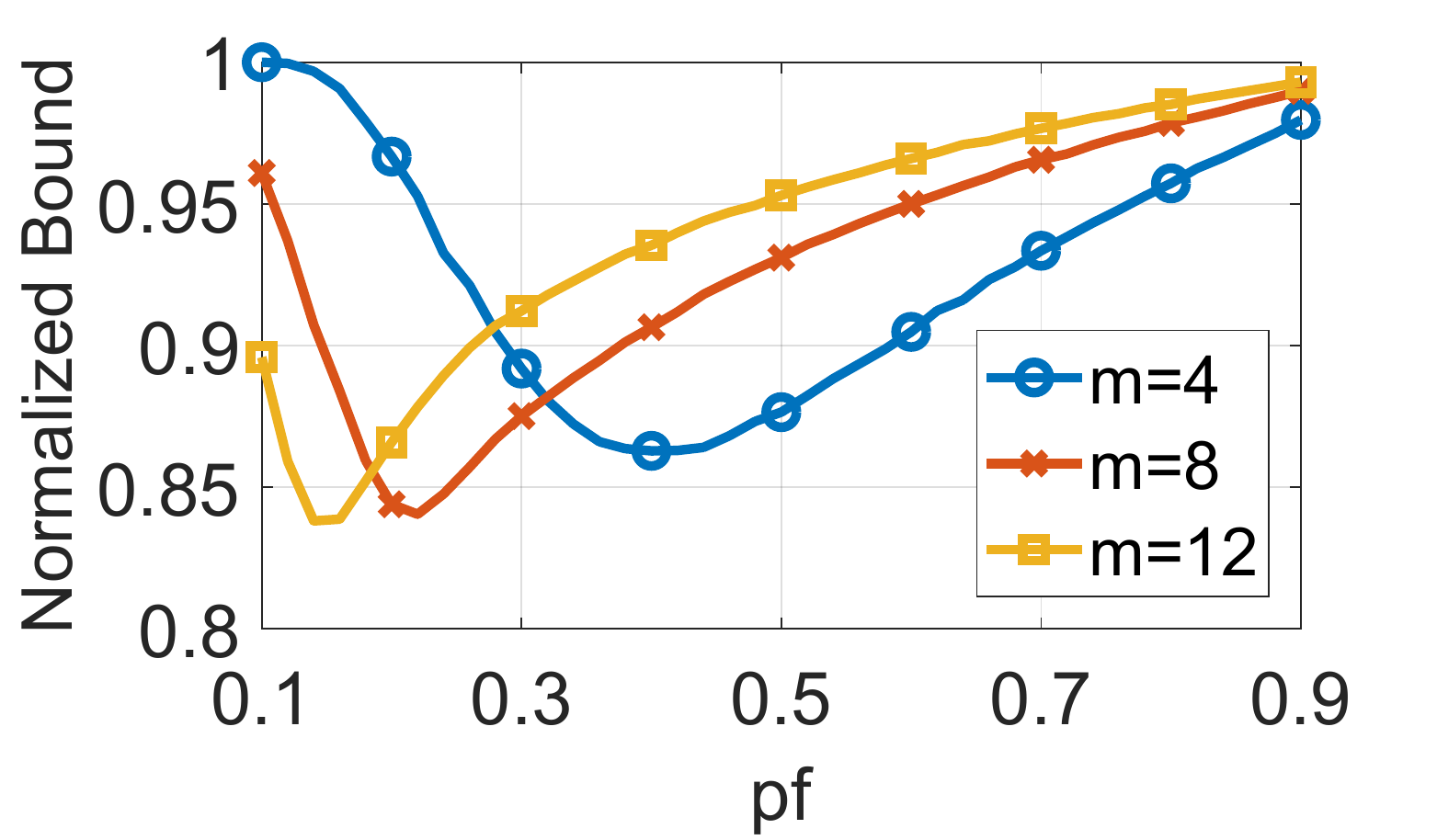}
    \label{fig:pf}
}
\hfil
\subfloat[vertex number]{
    \includegraphics[width=0.49\linewidth]{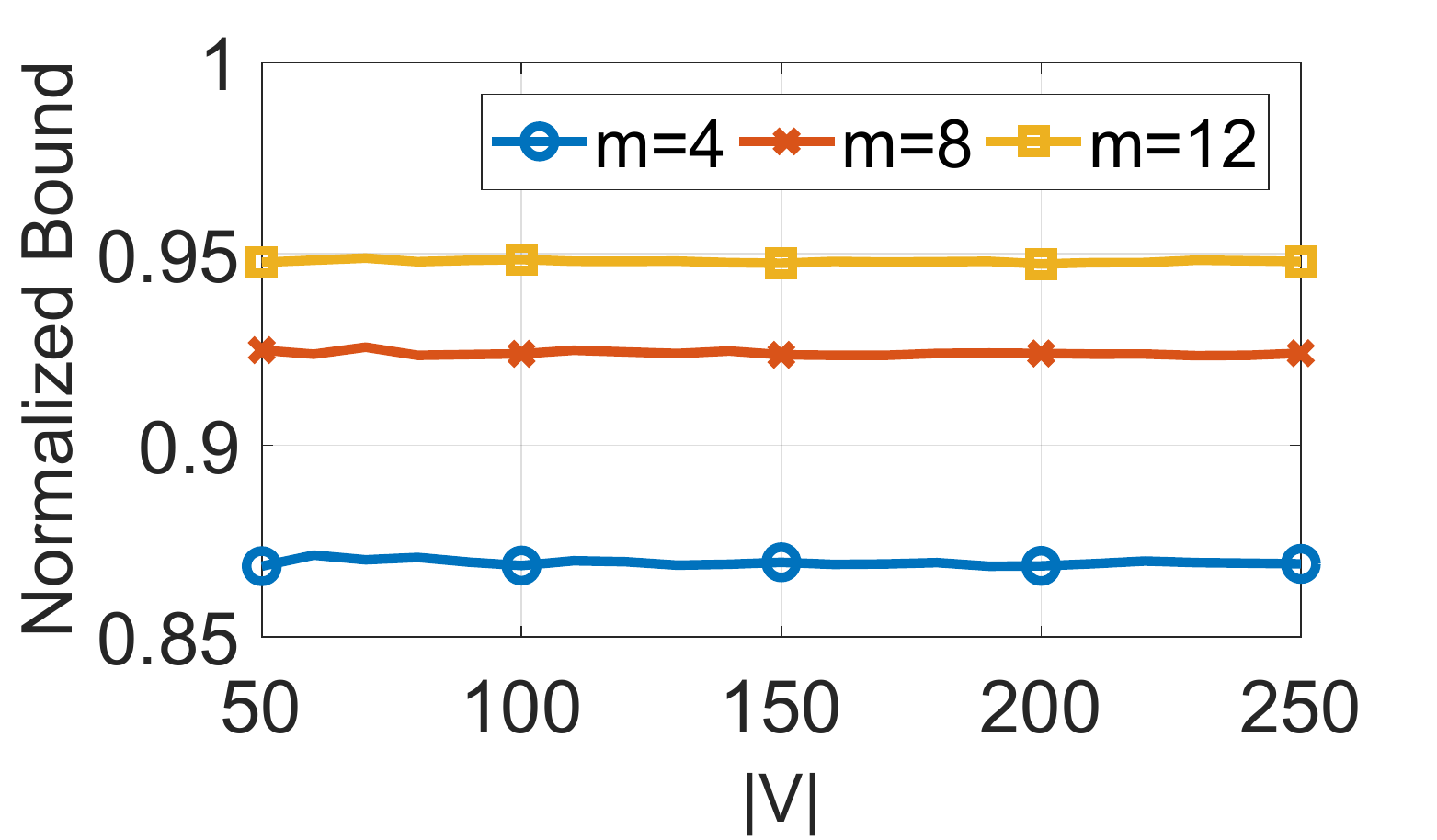}
    \label{fig:v}
}
\subfloat[deadline]{
    \includegraphics[width=0.49\linewidth]{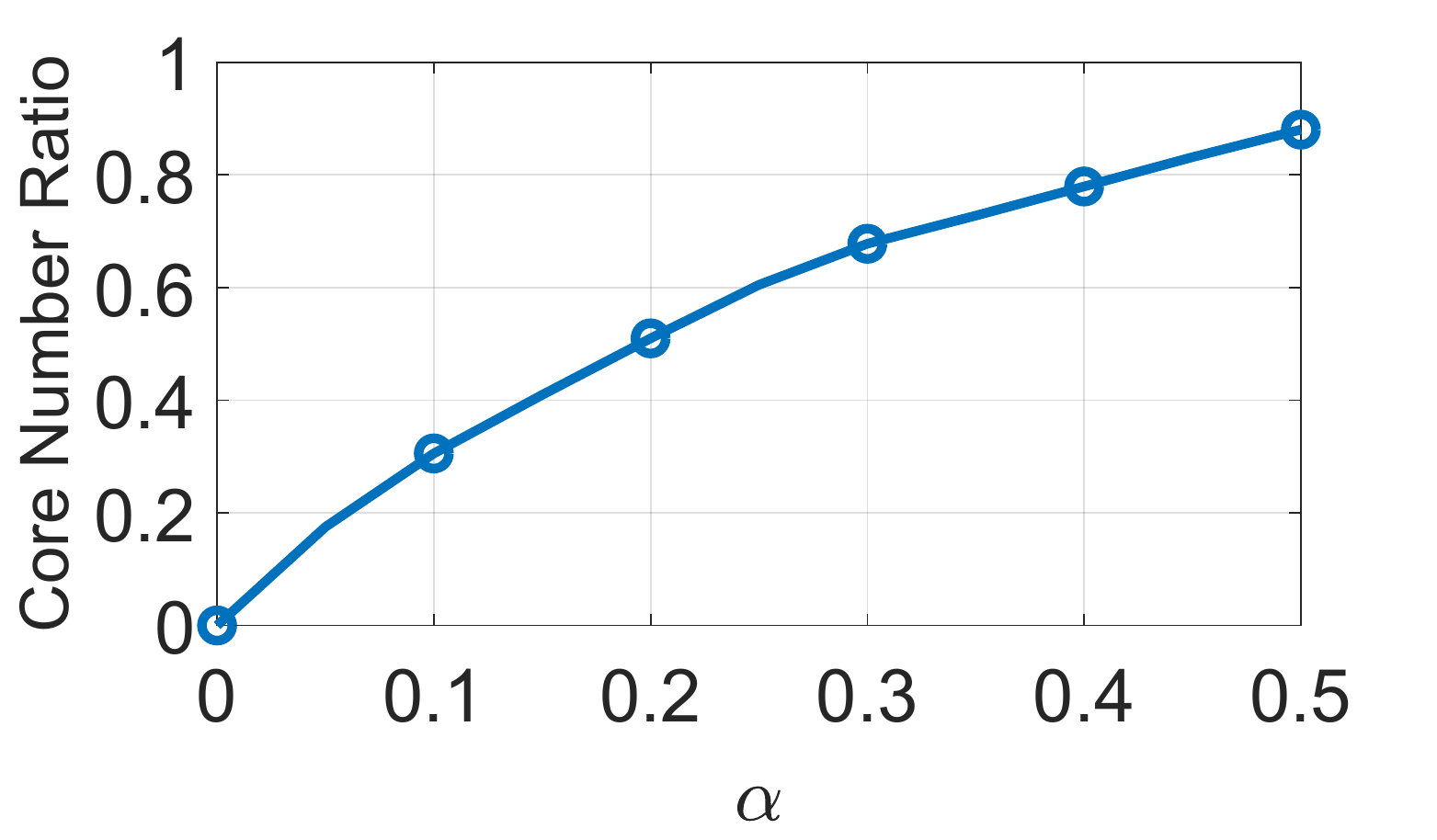}
    \label{fig:tm}
}
\hfil
\subfloat[parallelism factor]{
    \includegraphics[width=0.49\linewidth]{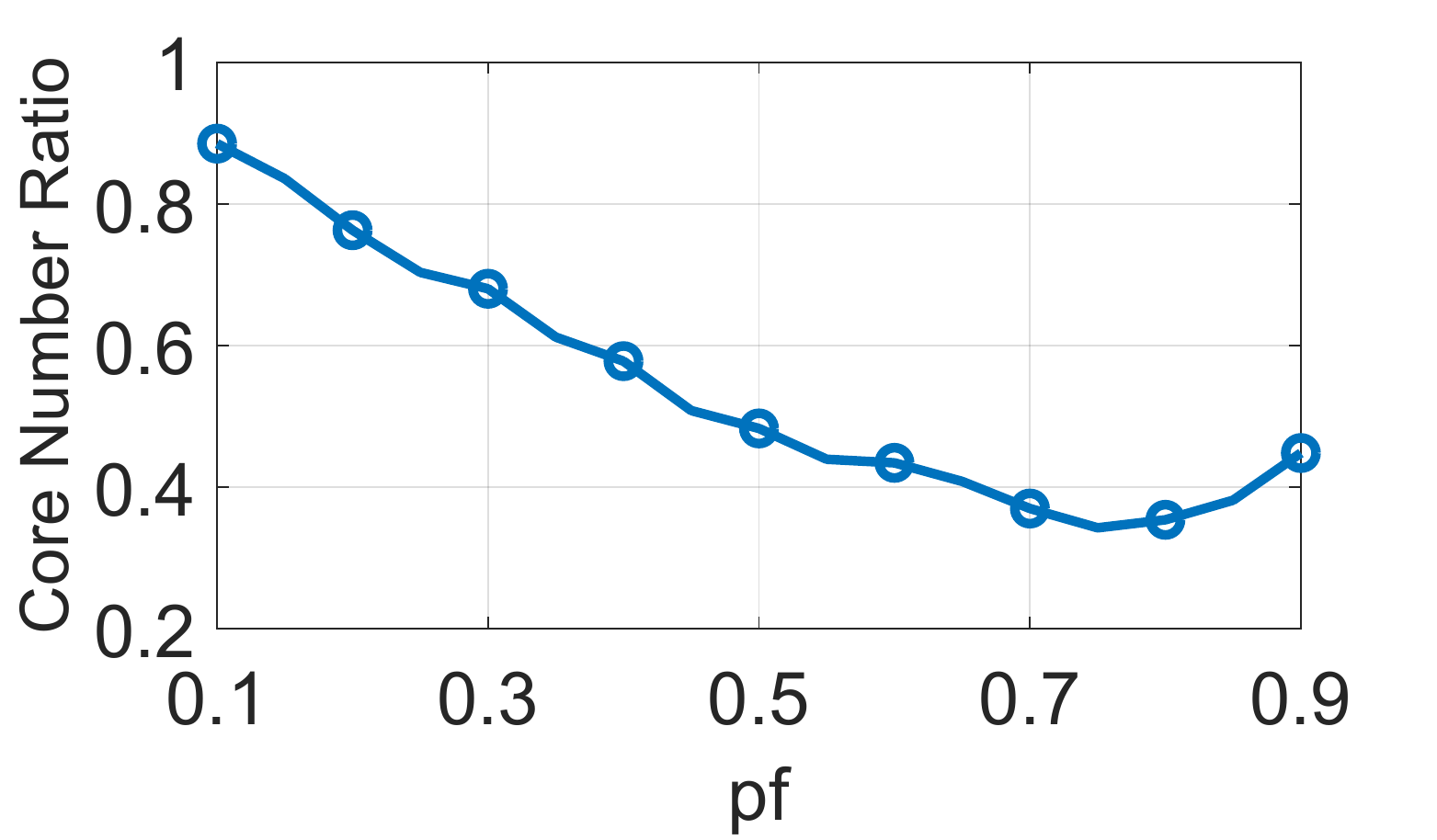}
    \label{fig:pm}
}
\subfloat[vertex number]{
    \includegraphics[width=0.49\linewidth]{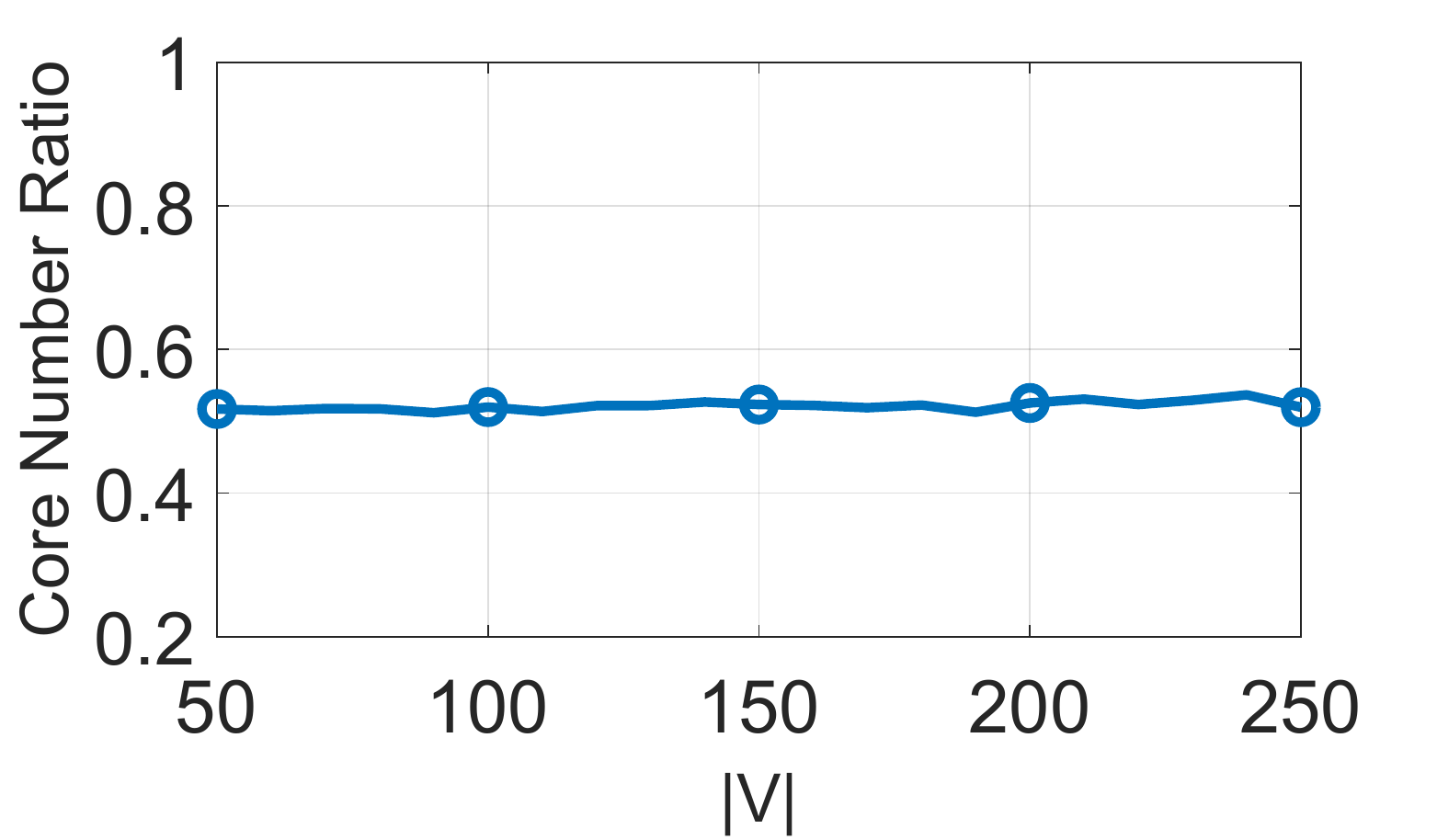}
    \label{fig:vm}
}
\caption{Evaluation of one DAG. Fig. \ref{fig:m}-\ref{fig:v} evaluate response time bound; Fig. \ref{fig:tm}-\ref{fig:vm} evaluate the core number required by a DAG task.}
\label{fig:evaluation_one}
\end{figure}

\begin{figure}[t]
\centering
\subfloat[core number as x-axis; under the same setting as Fig. \ref{fig:m}]{
    \includegraphics[width=0.78\linewidth]{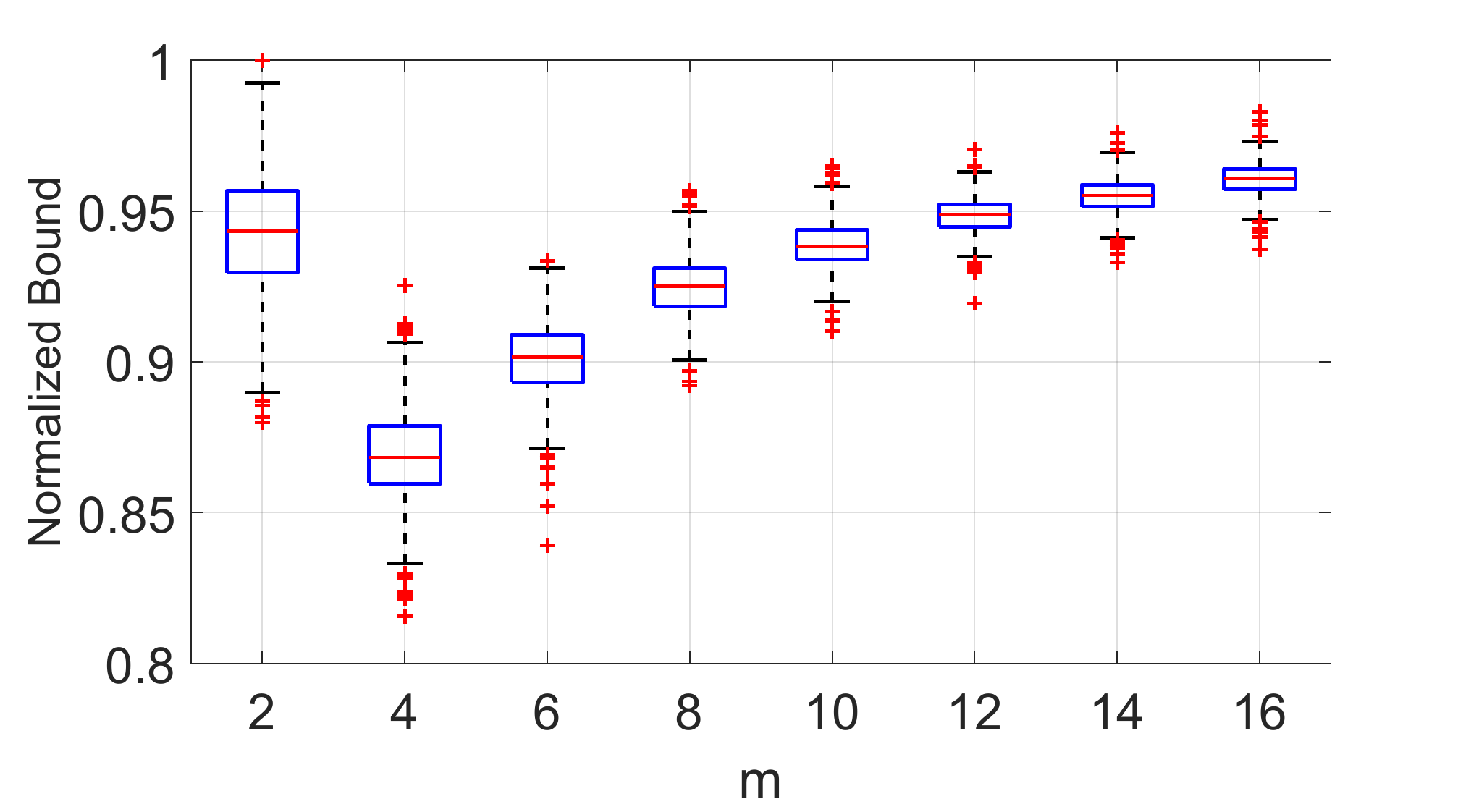}
    \label{fig:m_box}
}
\hfil
\subfloat[deadline as x-axis; under the same setting as Fig. \ref{fig:tm}]{
    \qquad
    \includegraphics[width=0.56\linewidth]{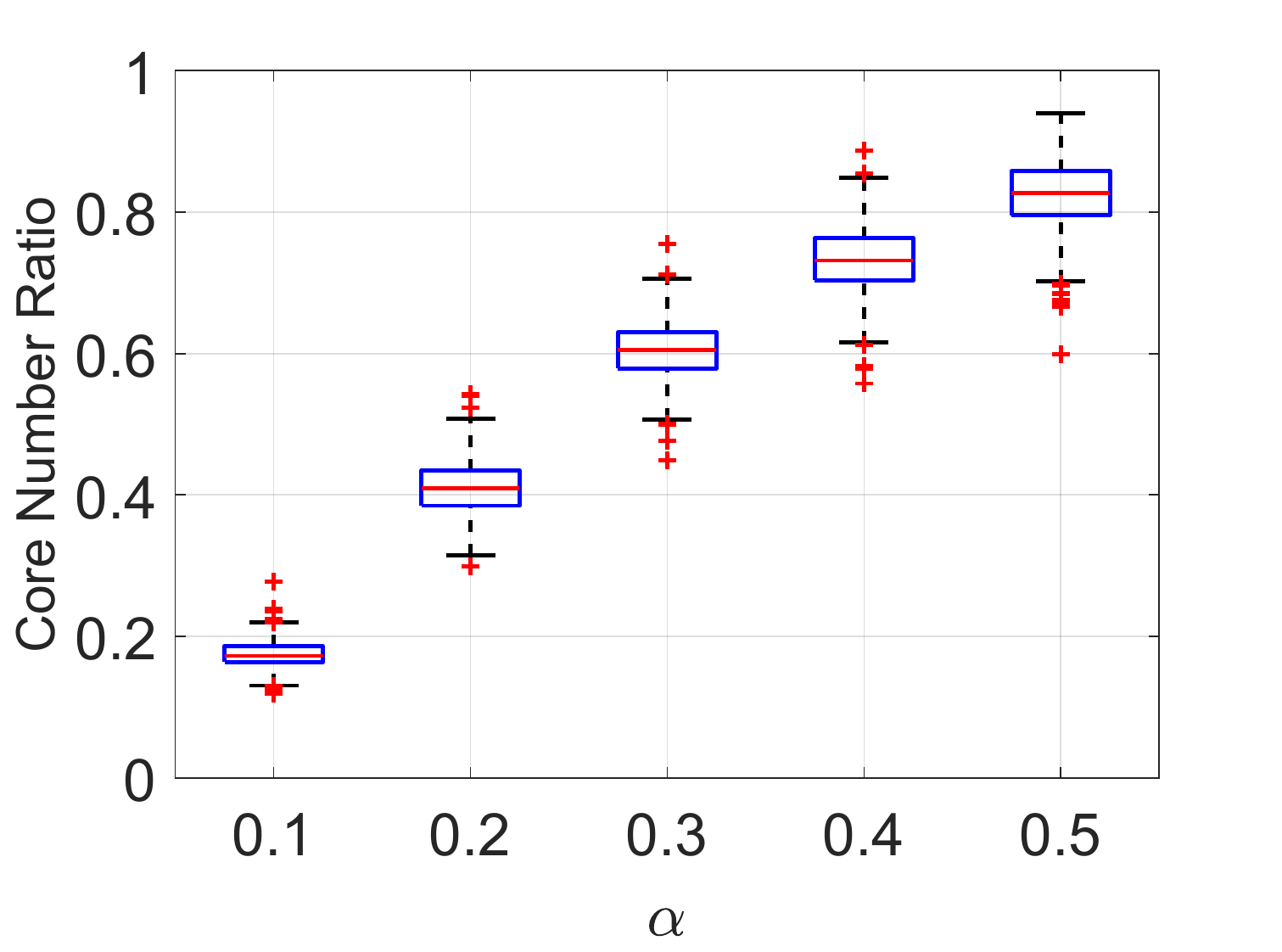}
    \label{fig:tm_box}
    \qquad
}
\caption{Performance variation regarding normalized bound and core number ratio.}
\label{fig:evaluation_box}
\end{figure}

This subsection evaluates the performance of our new bound with a single DAG task, compared to Graham's bound in (\ref{thm:classic_bound}) which, as stated in Section \ref{sec:related_work}, is the only result under the same setting as in this work.
Other related results \cite{he2019intra, zhao2020dag, he2021response, jiang2017semi, han2019response, voudouris2021bounding, melani2015response, sun2019calculating, sun2021calculating} all degrade to Graham's bound when considering a parallel real-time task modeled as a DAG under work-conserving scheduling on an identical multi-core platform.

\textbf{Task Generation.}
The DAG tasks are generated using the Erd\"os-R\'enyi method \cite{cordeiro2010random}, where the number of vertices $|V|$ is randomly chosen in a specified range. For each pair of vertices, it generates a random value in $[0, 1]$ and adds an edge to the graph if the generated value is less than a predefined \emph{parallelism factor} $\mathit{pf}$. The larger $\mathit{pf}$, the more sequential the graph is.
The period $T$ (which equals $D$ in the experiment) is computed by $L+\alpha(C-L)$, where $L$ is the length of the longest path, $C$ is the volume of the DAG and $\alpha$ is a parameter.
By (\ref{equ:original_core}), the number of cores required by a task is at most $\left \lceil \frac{1}{\alpha} \right \rceil$.
We consider $\alpha$ in [0, 0.5] to let heavy tasks require at least two cores.
The default settings are as follows.
The WCETs of vertices $c(v)$, the parallelism factor $\mathit{pf}$, the vertex number $|V|$ and $\alpha$ are randomly chosen in $[50, 100]$, $[0.1, 0.9]$, $[50, 250]$ and $[0, 0.5]$, respectively.
For each configuration (i.e., each data point in the figures), we randomly generate 5000 DAG tasks to compute the average value.

\textbf{Evaluation Using Normalized Bound.}
In this experiment, we use the normalized bound (i.e., the ratio between our bound and Graham's bound) as the metric for comparison. The smaller normalized bound, the large improvement our bound has.
The results are in Fig. \ref{fig:m}-\ref{fig:v}.
Fig. \ref{fig:m} shows the average normalized bound by changing the number of cores and Fig. \ref{fig:m_box} shows the variation of the normalized bound by using the box plot\footnote{
In a box plot, the middle line of the box indicates the median of data. The bottom and top edges of the box represent the 25th and 75th percentiles, respectively. The whiskers extending from the box show the range of data. And the outliers are plotted individually using the '+' symbol.
} under the same setting as Fig. \ref{fig:m}.
The improvement of our bound is up to 13.1\% with $m=4$ compared to Graham's bound.
As the core number becomes smaller and larger, our bound is closer to Graham's bound. This is because,
when the core number is small, both bounds approach $vol(G)$; when the core number becomes larger, both bounds approach $len(G)$.
Fig. \ref{fig:pf} shows the results by changing the parallelism factor.
The improvement is up to 16.2\% with $m=12$ and $\mathit{pf}=0.14$.
When $\mathit{pf}$ is small, i.e., the graph has high parallelism, our bound is closer to Graham's bound. This is because, for graphs with high parallelism, it is difficult to find a generalized path list with large volume and small $k$.
As $\mathit{pf}$ becomes larger, the graph is more sequential, and our bound becomes closer to Graham's bound (both bounds eventually approach $len(G)$).
The results with changing vertex number are presented in Fig. \ref{fig:v}, which shows that our analysis is insensitive to the vertex number of the graph.
By the data in Fig. \ref{fig:v}, our method reduces the response time bound by 13.1\% for $m=4$ on average compared to Graham's bound.

\textbf{Evaluation Using Core Number Ratio.}
For a DAG task, the number of required cores to satisfy its deadline can be computed.
The core number ratio is the ratio between the core number computed by (\ref{equ:new_core}) and (\ref{equ:original_core}).
The smaller core number ratio, the better our performance is.
To precisely compare the core numbers computed by different methods, the ceiling operation in (\ref{equ:new_core}) and (\ref{equ:original_core}) is not used in this experiment.
The results are in Fig. \ref{fig:tm}-\ref{fig:vm}.
Fig. \ref{fig:tm} shows the average core number ratio by changing $\alpha$ and Fig. \ref{fig:tm_box} shows the variation of the core number ratio by using the box plot under the same setting as Fig. \ref{fig:tm}.
Different $\alpha$ means different deadlines.
When $\alpha$ approaches 0, the deadline $D$ approaches $L$; the core number computed in (\ref{equ:original_core}) approaches infinite; the core number ratio approaches 0.
When $\alpha$ increases, the deadline $D$ becomes larger and close to $C$; both computed core numbers approach 1, and the core number ratio approaches 1.
Fig. \ref{fig:pm} shows the results by changing the parallelism factor $\mathit{pf}$.
When $\mathit{pf}$ increases, the DAG becomes more sequential which means that the volume $C$ becomes close to $L$.
Since $\alpha$ is randomly chosen in [0, 0.5] for all $\mathit{pf}$, in general, the deadline $D$ becomes close to $L$, which means that the core number computed by (\ref{equ:original_core}) can become large drastically.
Therefore, the core number ratio becomes smaller.
However, when $\mathit{pf}$ becomes close to 1, the generated DAG becomes extremely sequential, and we are into the corner cases. In the extreme case with $\mathit{pf}=1$, which means the DAG is a sequential task, we have $C=L=D$, and the task requires one core to meet its deadline. This explains why there is a trend of increase (a trend that the core number ratio becomes close to 1) for $\mathit{pf}>0.8$ in Fig. \ref{fig:pm}.
The results with changing vertex number are reported in Fig. \ref{fig:vm}, which shows that our method can reduce the number of cores by 47.9\% on average.
In summary, this experiment indicates that our method can significantly reduce the required core number for a DAG task.

\subsection{Evaluation of Multi-DAG Systems}
\label{sec:evaluation_set}

\begin{figure}[t]
  \centering
  \includegraphics[width=0.63\linewidth]{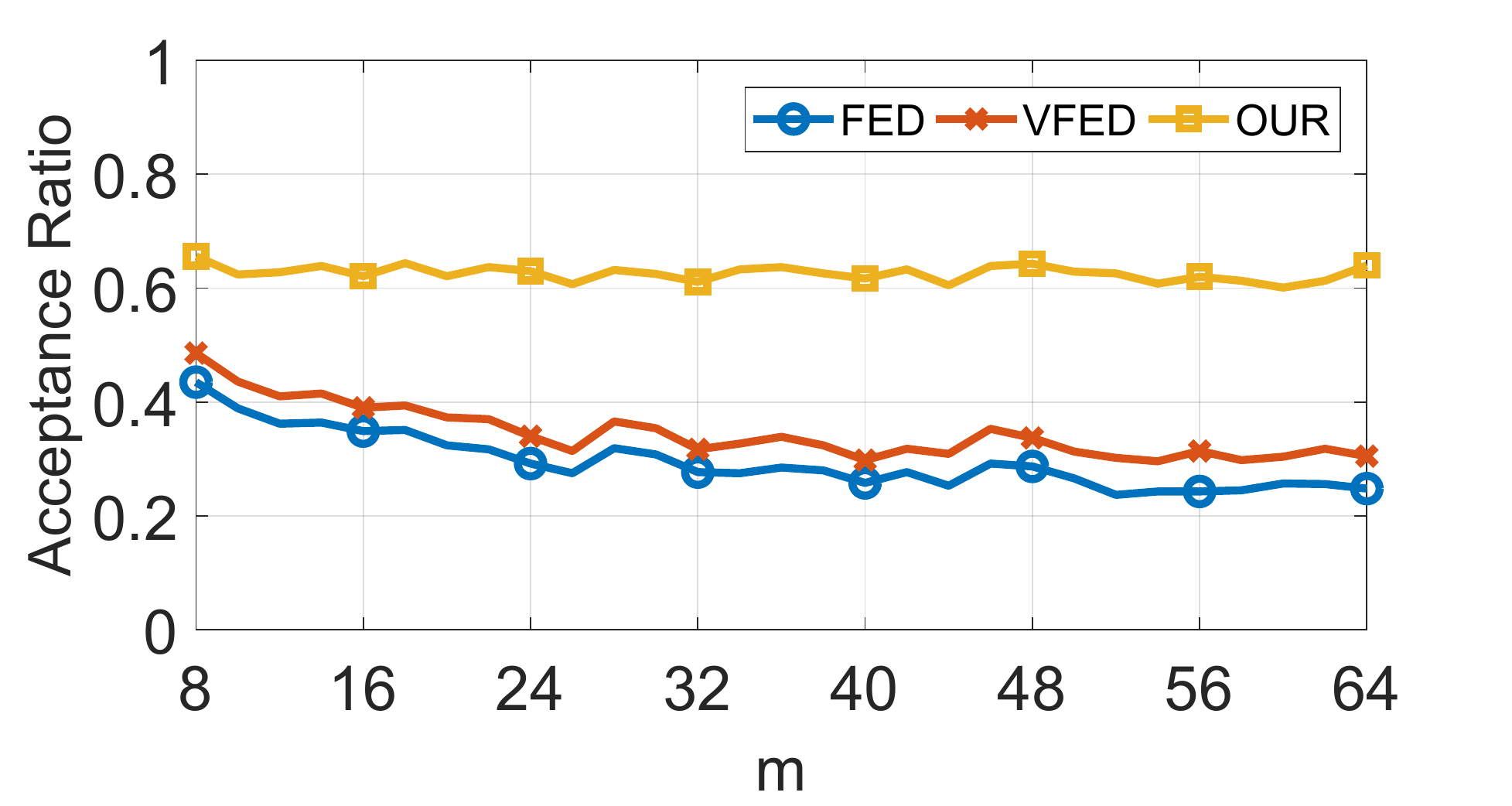}
  \caption{Evaluation of different numbers of cores.}
  \label{fig:g_m}
\end{figure}

\begin{figure}[t]
\centering
\subfloat[normalized utilization]{
    \includegraphics[width=0.49\linewidth]{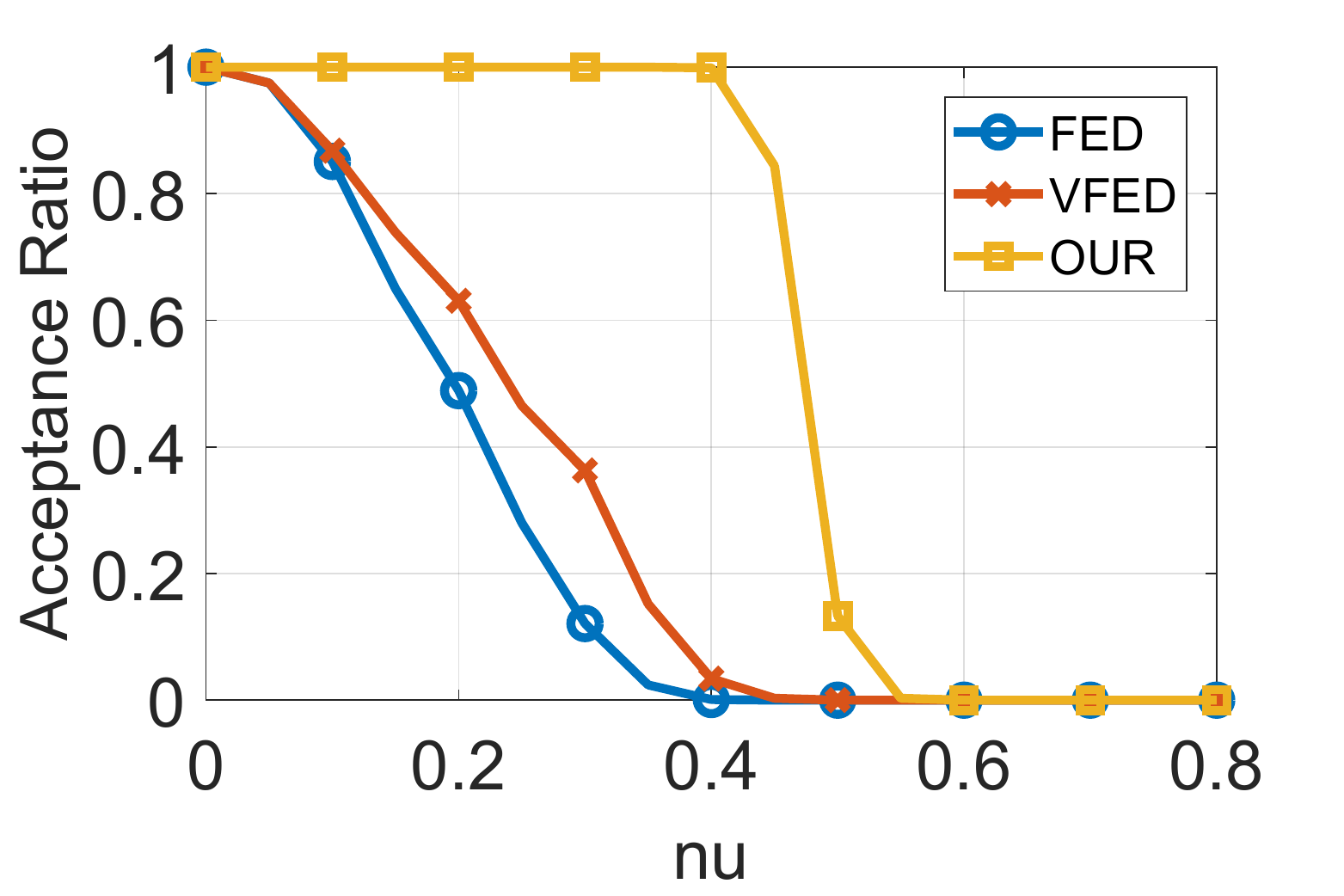}
    \label{fig:nu_set}
}
\subfloat[deadline]{
    \includegraphics[width=0.49\linewidth]{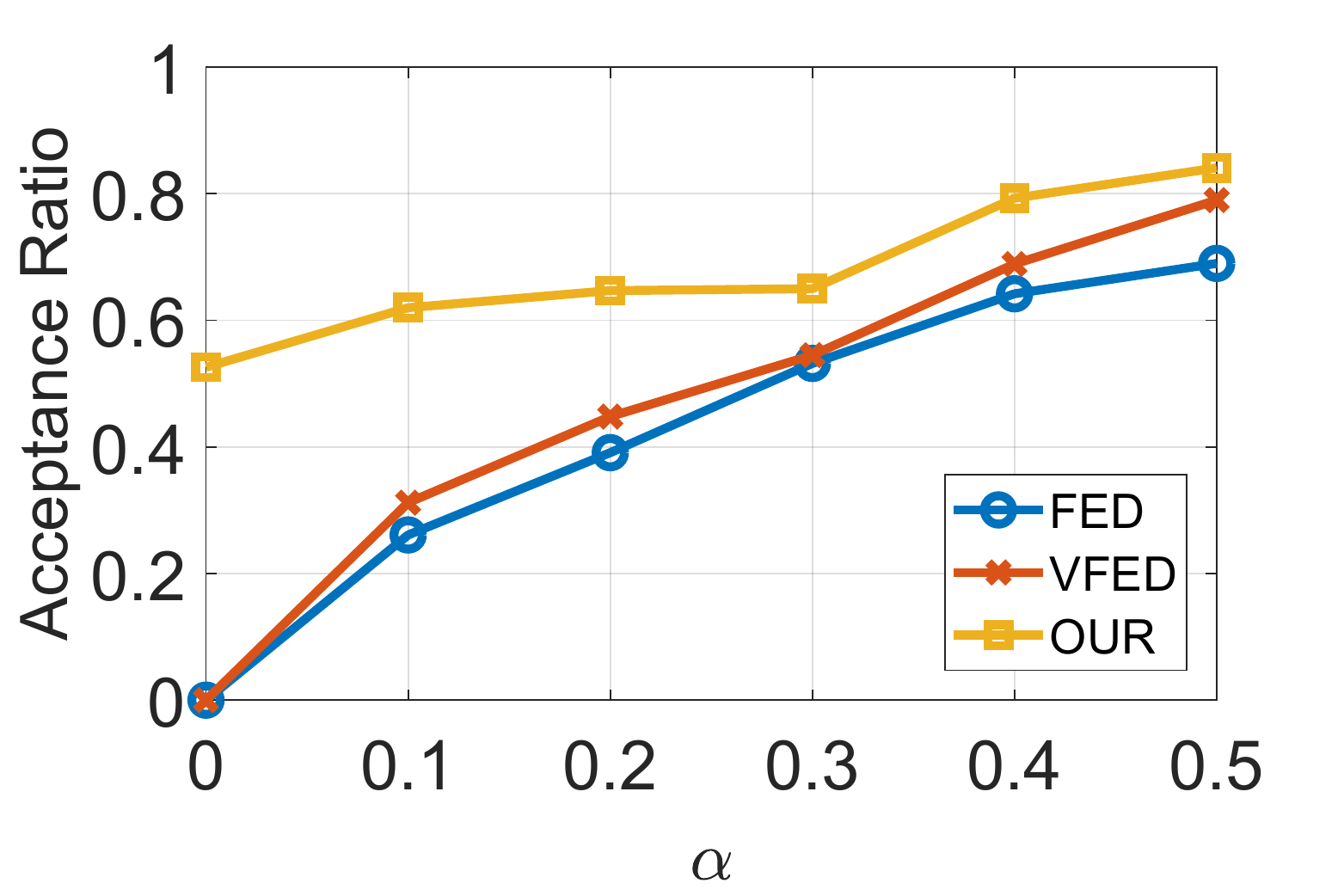}
    \label{fig:t_set}
}
\hfil
\subfloat[parallelism factor]{
    \includegraphics[width=0.49\linewidth]{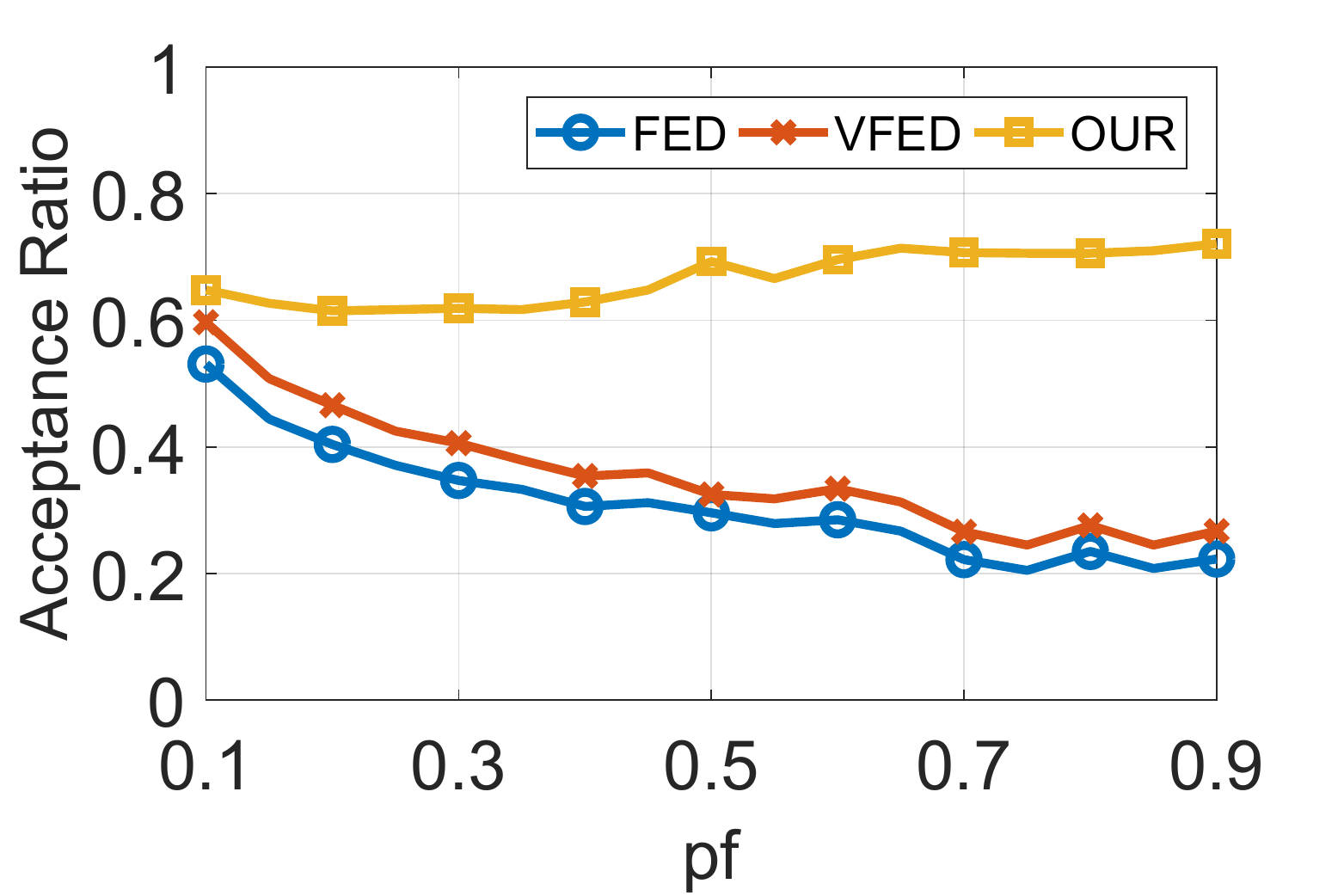}
    \label{fig:p_set}
}
\subfloat[vertex number]{
    \includegraphics[width=0.49\linewidth]{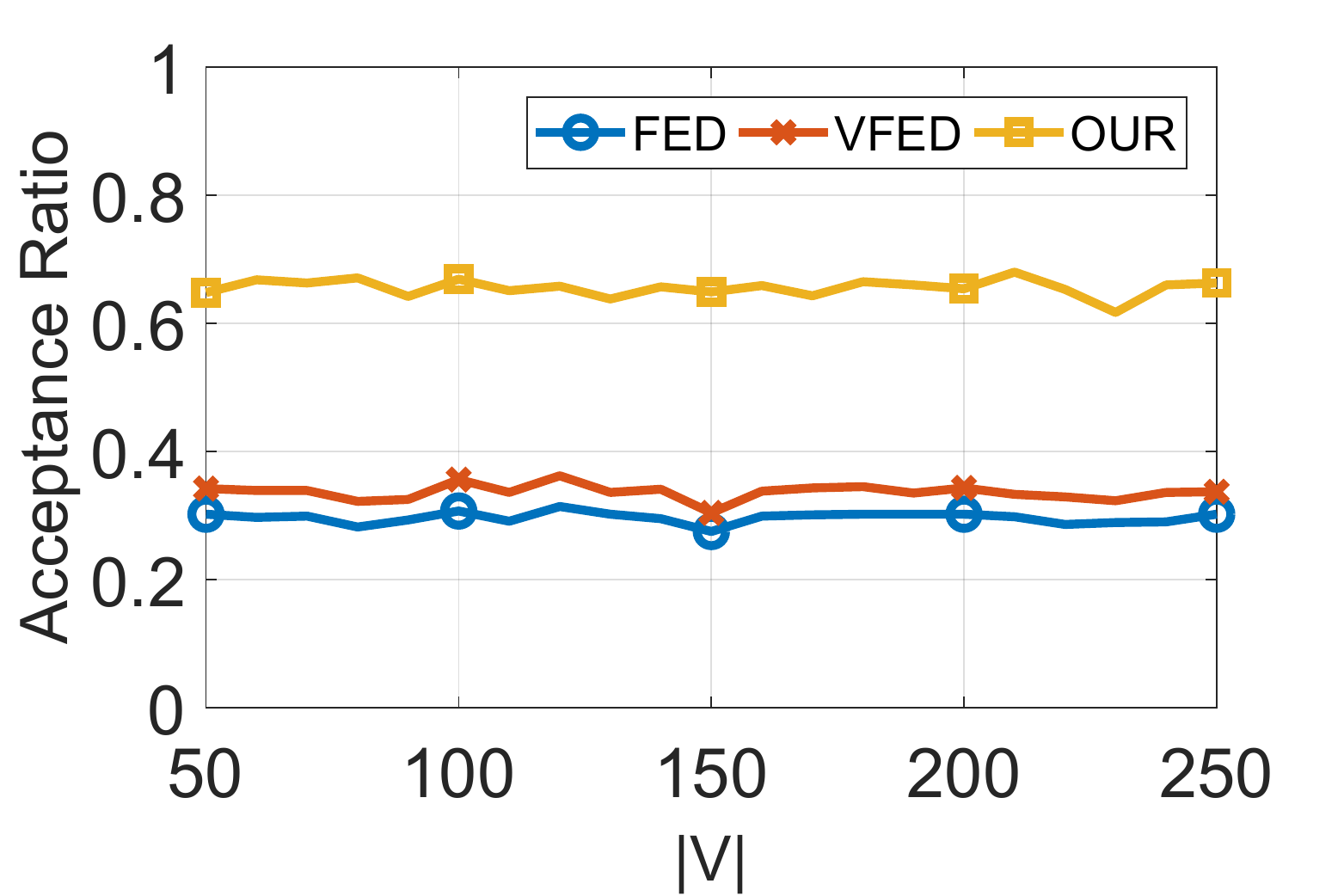}
    \label{fig:v_set}
}
\caption{Evaluation of multiple DAGs ($m$=32).}
\label{fig:evaluation_set}
\end{figure}

This subsection compares the following three methods for scheduling a DAG task set.
\begin{itemize}
  \item \textsf{OUR}. Our method presented in Section \ref{sec:extension}.
  \item \textsf{FED}. The original federated scheduling proposed in \cite{li2014analysis} based on Graham's bound.
  \item \textsf{VFED}. The virtually-federated scheduling in \cite{jiang2021virtually}, by
  adding servers on top of federated scheduling (based on Graham's bound) to reclaim
  unused processing capacity.
\end{itemize}

As shown in \cite{jiang2021virtually}, \textsf{VFED} has the best performance among all existing multi-DAG scheduling algorithms of different paradigms (federated, global, and partitioned), so we only include \textsf{VFED} in our comparison. Since \textsf{VFED} is based on Graham's bound, there is a potential to achieve even better schedulability by integrating our new bound with the idea of \textsf{VFED}, which will be studied in our future work.

\textbf{Task Set Generation.}
DAG tasks are generated by the same method as Section \ref{sec:evaluation_one} with $c(v)$, $\mathit{pf}$, $|V|$, $\alpha$ randomly chosen in [50, 100], [0.1, 0.9], [50, 250], [0, 0.5], respectively.
The number of cores $m$ is set to be 32 (but changing in Fig. \ref{fig:g_m}) and the normalized utilization $\mathit{nu}$ of task sets is randomly chosen in [0, 0.8].
To generate a task set with specific utilization, we randomly generate a DAG task and add it to the task set until the total utilization reaches the required value.
For each configuration (i.e., each data point in the figures), we randomly generate 5000 task sets to compute the average acceptance ratio.

We first evaluate the schedulability of task sets using acceptance ratio as the metric when scheduling on different core numbers. Fig. \ref{fig:g_m} presents the result. The larger acceptance ratio, the better the performance.
Our method can improve the system schedulability by 33.5\% for $m=64$ compared to \textsf{VFED}.
Since all three methods are of the federated scheduling paradigm, their performances are generally capable of scaling to the increase of the number of cores.
However, with the number of cores increasing, the performance of \textsf{FED} and \textsf{VFED} slightly decrease while our scheduling method does not.
This is because by utilizing the information of multiple long paths, our methods can better make use of the computing powers provided by the multi-core platform.
In the following experiments, we use the core number $m=32$ as a representative for evaluation.

We second compare the acceptance ratio under different settings for $m=32$ and the results are reported in Fig. \ref{fig:evaluation_set}.
In Fig. \ref{fig:nu_set}, compared to \textsf{VFED}, the improvement of acceptance ratio is up to 96.5\% with $\mathit{nu}=0.4$.
Consisting with Fig. \ref{fig:tm}-\ref{fig:vm}, Fig. \ref{fig:t_set}-\ref{fig:v_set} shows similar trends, and the reasons of these trends are the same with Fig. \ref{fig:tm}-\ref{fig:vm}.
Fig. \ref{fig:t_set} shows that our method consistently outperforms other methods.
In Fig. \ref{fig:p_set}, compared to \textsf{VFED}, the improvement of acceptance ratio is up to 46.5\% with $\mathit{pf}=0.85$.
Fig. \ref{fig:v_set} shows that compared to \textsf{VFED}, the average improvement is 31.9\% for $m$=32.
This experiment shows that the proposed method consistently outperforms the original federated scheduling and the state-of-the-art scheduling techniques for multi-DAG task systems by a considerable margin.